\documentclass[a4paper,USenglish,cleveref, autoref, thm-restate]{socg-lipics-v2021}
\pdfoutput=1 
\hideLIPIcs  

\graphicspath{{./graphics/}}

\title{Fare Zone Assignment on Trees}

\author{Martin Hoefer}{RWTH Aachen University, Germany}{mhoefer@cs.rwth-aachen.de}{https://orcid.org/0000-0003-0131-5605}{Supported by DFG Research Unit ADYN (project number 411362735) and grant DFG Ho 3831/9-1 (project number 514505843).}
\author{Lennart Kauther}{RWTH Aachen University, Germany}{kauther@oms.rwth-aachen.de}{https://orcid.org/0000-0002-5791-9195}{}
\author{Philipp Pabst}{RWTH Aachen University, Germany}{pabst@oms.rwth-aachen.de}{https://orcid.org/0009-0002-6417-6876}{Supported by the DFG Research Training Group UnRAVeL (project number 282652900).}
\author{Britta Peis}{RWTH Aachen University, Germany}{peis@oms.rwth-aachen.de}{https://orcid.org/0000-0002-8938-8843}{}
\author{{Khai Van} Tran}{RWTH Aachen University, Germany}{tran@oms.rwth-aachen.de}{https://orcid.org/0000-0002-8277-7010}{}

\authorrunning{M.\,Hoefer, L.\,Kauther, P.\,Pabst, B.\,Peis and K.\,V.\,Tran} 

\Copyright{Martin Hoefer, Lennart Kauther, Philipp Pabst, Britta Peis and Khai Van Tran} 

\ccsdesc{Mathematics of computing~Combinatorial optimization}
\ccsdesc{Theory of computation~Approximation algorithms analysis}
\ccsdesc{Theory of computation~Theory and algorithms for application domains}

\keywords{Fare Zone Design, Approximation Algorithms, Network Pricing, Tollbooth Problem}

\category{} 

\acknowledgements{The authors would like to thank Sven Müller for fruitful discussions and inspiration regarding the problem.}

\nolinenumbers 

\EventEditors{John Q. Open and Joan R. Access}
\EventNoEds{2}
\EventLongTitle{42nd Conference on Very Important Topics (CVIT 2016)}
\EventShortTitle{CVIT 2016}
\EventAcronym{CVIT}
\EventYear{2016}
\EventDate{December 24--27, 2016}
\EventLocation{Little Whinging, United Kingdom}
\EventLogo{}
\SeriesVolume{42}
\ArticleNo{23}

\usepackage[nolist]{acronym}
\usepackage[
shortend,
boxed,
vlined
]{algorithm2e}
\usepackage{anyfontsize}
\usepackage{bm}
\usepackage{booktabs}
\usepackage{nicefrac}
\usepackage{float}
\usepackage{mathtools}
\usepackage[defaultlines=2, all]{nowidow}
\usepackage[nodisplayskipstretch]{setspace}
\usepackage{tabularx}
\usepackage{tikz}
\usepackage{xcolor}
\usepackage{todonotes}
\usepackage[section]{placeins}
\usepackage{xurl} 

\graphicspath{{graphics/}}

\usetikzlibrary{calc, positioning, shapes.geometric, shapes.misc, intersections, arrows, arrows.meta, overlay-beamer-styles, decorations.pathreplacing}
\newcommand{\tikzdots}{\hspace*{1.75pt}$\dots$}
\input{TikzSet.tex}
\pgfdeclarelayer{bg}    
\pgfdeclarelayer{fg}    
\pgfsetlayers{bg,main,fg}  
\definecolor{GreyClass2}{RGB}{189,189,189}
\definecolor{graphcolorpeach}{RGB}{254,232,200}
\definecolor{graphcolorpeach-orange}{RGB}{253,187,132}
\definecolor{graphcolordark-orange}{RGB}{227,74,51}
\definecolor{graphcolorred}{RGB}{215,25,28}
\definecolor{graphcoloralternative-peach-orange}{RGB}{253,174,97}
\definecolor{graphcoloralternative-peach-orange2}{RGB}{244,165,130}
\definecolor{battleshipgrey}{rgb}{0.52, 0.52, 0.51}
\definecolor{rwthblue}{RGB}{0,84,159}
\definecolor{rwthlightblue}{RGB}{142,189,229}
\definecolor{rwthblue_75prcnt}{RGB}{64,127,183}
\definecolor{rwthlightblue}{RGB}{142,189,229}
\definecolor{rwthblue_25prcnt}{RGB}{199,221,242}
\definecolor{rwthgold}{HTML}{F9A02C}
\definecolor{rwthsilver}{HTML}{797A7C}
\definecolor{rwthmagenta}{RGB}{227,0,102}
\definecolor{rwthyellow}{RGB}{255,237,0}
\definecolor{rwthpetrol}{RGB}{0,97,101}
\definecolor{rwthcyan}{RGB}{0,152,161}
\definecolor{rwthgreen}{RGB}{87,171,39}
\colorlet{rwthgreen50prcnt}{rwthgreen!50}
\definecolor{rwthmay}{RGB}{189,205,0}
\colorlet{rwthmaygreen50prcnt}{rwthmay!50}
\definecolor{rwthorange}{RGB}{246,168,0}
\definecolor{rwthred}{RGB}{204,7,30}
\colorlet{rwthred_50prcnt}{rwthred!50}
\definecolor{rwthbordeaux}{RGB}{161,16,53}
\definecolor{rwthviolett}{RGB}{97,33,88}
\definecolor{rwthpurple}{RGB}{122,111,172}
\newcommand{\ie}{i.e.,\ } 

\def\hmath$#1${\texorpdfstring{{\rmfamily\textit{#1}}}{#1}}

\newcommand{\rom}[1]{\mathrm{\uppercase\expandafter{\romannumeral #1\relax}}}
\let\emptyset\varnothing
\let\epsilon\varepsilon
\let\rho\varrho
\let\phi\varphi

\DeclareMathOperator*{\argmax}{arg\,max}
\DeclareMathOperator*{\argmin}{arg\,min}
\renewcommand{\tilde}{\widetilde}

\newcommand{\pathGraph}{P}

\DeclareMathOperator{\OPT}{OPT}
\DeclareMathOperator{\ALG}{ALG}
\DeclareMathOperator{\DP}{DP}
\DeclareMathOperator{\XP}{\textsf{XP}}
\DeclareMathOperator{\FPT}{\textsf{FPT}}

\newcommand{\ub}{u}
\newcommand{\weight}{w}

\newcommand{\sink}{t}

\newcommand{\DPTable}{R}
\newcommand{\cutset}{F}
\newcommand{\clause}{C}
\newcommand{\eStart}{\theta}

\newcommand{\varvertex}{x}

\DeclareMathOperator{\con}{cong}

\newcommand{\rev}[1]{
	\operatorname{rev}_{#1}
}

\newcommand{\blockA}{($A$)}
\newcommand{\blockB}{($B$)}
\newcommand{\blockC}{($C$)}

\newcolumntype{Y}{>{\centering\arraybackslash}X}
\newcolumntype{L}[1]{>{\raggedright\arraybackslash}p{#1}} 
\newcolumntype{C}[1]{>{\centering\arraybackslash}p{#1}} 
\newcolumntype{R}[1]{>{\raggedleft\arraybackslash}p{#1}} 

\makeatletter
\patchcmd{\@algocf@start}
{-1.5em}
{-.75em}
{}{}
\makeatother

\newcommand{\algorithmFont}[1]{\small\ttfamily{#1}}

\SetAlCapSkip{1ex}
\SetAlFnt{\algorithmFont}
\SetNlSkip{1.65em}
\newcommand{\algorithmInit}{
	\setstretch{1.05}
	\DontPrintSemicolon%
	\SetNlSty{lineNumberFont}{}{}
	\SetKwProg{Fn}{function}{}{end}%
	\SetArgSty{argFont} 
	\SetKwFunction{KwFn}{function} 
	\SetKwSty{keywordFont}%
	\SetCommentSty{commentFont}%
	\SetKw{break}{break}%
	\SetInd{.3em}{.6em}
}
\SetKw{To}{to}
\SetKw{With}{with}
\SetKw{And}{and}
\SetKw{Or}{or}
\SetKw{Not}{not}
\SetKw{Downto}{downto}
\SetKw{DownTo}{downto}
\SetKw{Step}{step}
\SetKwRepeat{Do}{do}{while}

\let\oldnl\nl
\newcommand{\nonl}{\renewcommand{\nl}{\let\nl\oldnl}}

\usepackage{thm-restate}

\renewcommand{\paragraph}[1]{\medskip \noindent \textbf{\textsf{#1}} $\,$}

\begin{document}
	
	\begin{acronym}
	\acro{FZA}[FZA]{\textsc{Fare Zone Assignment}}
	\acro{Max2SAT}[Max2- SAT]{\textsc{Maximum 2-Satisfiability}}
\end{acronym}

	\maketitle
	
	\begin{abstract}
		Designing fare systems for public transportation networks is a challenging task. A popular approach is to partition the network into fare zones (``zoning'') and fix journey prices depending on the number of traversed zones (``pricing'').
		In this paper, we focus on finding revenue-optimal solutions to the zoning problem for a given subadditive pricing function. 
		We consider tree networks with $n$ vertices, since trees already pose non-trivial algorithmic challenges. Our main results are efficient algorithms that yield a simple $\mathcal{O}(\log n)$-approximation as well as a more involved $\mathcal{O}(\log n/\log \log n)$-approxi\-ma\-tion. We show that rooted instances, in which all demand arises at a single source, can be solved exactly. We further show \textsf{APX}-hardness for general instances on star graphs. For paths, we prove strong \textsf{NP}-hardness and outline a PTAS. Moreover, we show that computing an optimal solution is in \textsf{FPT} or \textsf{XP} for several natural problem parameters.
	\end{abstract}
	
	\section{Introduction}
	\label{sec:Introduction}
	Designing optimal fare structures is a complex challenge for public transport operators. While passenger demand comprises diverse streams with varying budgets and willingness to pay, operators strive for simple and transparent pricing schemes. Flat fare systems, where all passengers pay a uniform rate regardless of distance, offer maximum simplicity but disproportionately penalize short trips and inadvertently subsidize longer journeys. Conversely, purely continuous distance-based fares address this unfairness but introduce high transactional complexity. As a practical compromise, operators frequently divide the transport network into discrete fare zones, charging passengers based on the number of traversed zones (as implemented in metropolitan networks like Berlin, London, or Montréal). In this paper, we focus on such zone-based strategies, in particular, on the \emph{connected zones} model.
	Implementing zone-based fares involves two highly interdependent decisions:
	\begin{enumerate}[(i)]
		\item Defining a spatial partitioning of the transport network into fare zones that are internally connected, i.e., each zone forms a connected subgraph (``zoning'').
		\item Specifying a pricing function that maps a journey to a ticket cost based on the traversed zones (``pricing'').
	\end{enumerate}
	In this paper, we restrict ourselves to the zoning subproblem by treating the pricing function as \textit{fixed input}. 
	Specifically, we study a counting-zones model with a given pricing function $f\colon \mathbb{N}_0 \rightarrow \mathbb{R}_{\geq 0}$ that maps the number of zone border crossings to a ticket price. We assume that $f$ is non-decreasing and subadditive (that is, $f(x+y) \leq f(x) + f(y)$ for all $x,y \in \mathbb{N}$). Here, monotonicity reflects the fact that passengers should not pay less for longer journeys, and subadditivity ensures that passengers cannot decrease their fare by splitting their journey into several subpaths. In \cite{Schoebel2025}, these two properties are termed the \emph{no-elongation property} and the \emph{no-stopover property}, respectively.
	
	Operating a public transport system requires balancing social utility with financial viability, as fare revenues must recover a significant portion of operating costs. Revenue-optimal zoning can thus be viewed as a building block towards more balanced designs that satisfy practical constraints of real-world operations. The goal of this paper is to understand the theoretical foundations of revenue-optimal spatial zone partitioning. We restrict our considerations to tree networks as this allows for a more convenient problem definition and, as we show, the zoning problem restricted to trees already reaches the limits of computational tractability. Notably, many public transit systems---particularly commuter rail or metropolitan networks---can be reduced to tree-like structures by consolidating vertices within densely interconnected city centers \cite{Angeloudis2006,Deribble2012,Deribble2010}. 
	Observe that on trees, each journey can be specified by a path, and a partitioning of the network into zones corresponds to a selection of edges which ``cut'' the transit network into connected subtrees (representing the zones).
	
	\paragraph{The model.} Given an integer $\ell$, we abbreviate $[\ell] \coloneqq \{1, \ldots, \ell\}$. We model the transportation network as a tree $T=(V,E)$ with $|V| = n$, and assume that a fixed, non-decreasing, subadditive pricing function $f\colon \mathbb{N}_0 \rightarrow \mathbb{R}_{\geq 0}$ is given. 
	Each type of traveler $i\in [k]$ (from now on called \emph{commodity}) is specified by a path $\pathGraph_i\subseteq E$, a weight $\weight_i\in \mathbb{N}$ representing the number of travelers of this type, and a zone-transition budget $\ub_i\in \mathbb{N}_0$. This budget $\ub_i$ can be derived from the travelers' maximum willingness to pay evaluated against the pricing function $f$; it represents the maximum number of zone border crossings (i.e., edge cuts) along its path $\pathGraph_i$ that travelers of commodity $i$ will tolerate. We say that $i$ \emph{drops out}, e.g.,\ by switching to an alternative mode of transport, if strictly more than $\ub_i$ of the edges on $\pathGraph_i$ are cut. Otherwise, if at most $\ub_i$ edges of $\pathGraph_i$ are cut, we say commodity $i$ is \emph{served}. The objective is to select a set of edge cuts that maximizes the total revenue extracted from all served commodities.
	
	\begin{definition}[\textsc{Fare Zone Assignment (FZA)}]
		\label{def:FZA_formalDef}
		Given a tree $T=(V,E)$, commodities $(\pathGraph_i, \ub_i, \weight_i)_{i \in [k]}$, and a non-decreasing, subadditive pricing  function $f\colon {\mathbb{N}_0 \rightarrow \mathbb{R}_{\geq 0}}$, the problem \textsc{Fare Zone Assignment (FZA)} asks for a subset $\cutset \subseteq E$ maximizing $\sum_{i \in [k]} \rev{i}(\cutset)$, where
		\[
		\rev{i}(\cutset)= 
		\begin{cases}
			\weight_i \cdot f(|\pathGraph_i \cap \cutset|) & \text{ if } \; |\pathGraph_i \cap \cutset| \leq \ub_i,\\
			0 & \text{ otherwise.}
		\end{cases}
		\]
	\end{definition}
	\acused{FZA}
	Note that since we restrict the problem definition to trees, it suffices to specify a subset of cut edges $F\subseteq E$ to uniquely determine the partition into fare zones. For $M \subseteq [k]$ and $\cutset\subseteq E$ we use the notation $\rev{M}(F) = \sum_{i \in M} \rev{i}(F)$. We explicitly allow distinct commodities $i\neq j$ to share the same path $\pathGraph_i = \pathGraph_j$, however, we then assume $\ub_i \neq \ub_j$. Otherwise, we merge $i$ and $j$ into a single commodity with weight $\weight_i + \weight_j$. Thus, the maximum number of distinct commodities $k$ is bounded by $n^3$. We permit $f(0) > 0$, reflecting base fares charged to customers whose journeys remain entirely within a single zone.
	Finally, we call an instance of \ac{FZA} \textit{rooted} if the paths of all commodities $i \in [k]$ share a common endpoint (root). 
	
	Consider the instance of \ac{FZA} illustrated in \cref{fig:FZA-example} with pricing function $f(x) = x$ and $k=5$ unit weight commodities with capacities $u=(1,1,3,3,5)$. The  optimal cut set (indicated by the curved lines) serves commodities 1 to 4 and obtains a revenue of $8$.
	
	\begin{figure}[tbhp]
		\centering
		\begin{tikzpicture}[
    every node/.style={draw = none},
    vertex/.style={circle, fill=black, inner sep=1.5pt, minimum size=3pt, font=\small, text=white},
    pathstyle/.style={draw, line width=1.8pt, black},
    rwthmagentapathstyle/.style={draw, line width=1.2pt, rwthmagenta},
    rwthlightbluepathstyle/.style={draw, line width=1.2pt, rwthlightblue},
    rwthorangepathstyle/.style={draw, line width=1.2pt, rwthorange},
    rwthbluepathstyle/.style={draw, line width=1.2pt, rwthblue},
    rwthgreypathstyle/.style={draw, line width=1.2pt, rwthsilver},
    rwthbluepathstyle/.style={draw, line width=1.2pt, rwthblue},
    rwthbordeauxpathstyle/.style={draw, line width=1.2pt, rwthbordeaux},
    scale=0.725 
    ]
    
    \pgfmathsetmacro{\Ysep}{1} 
    \pgfmathsetmacro{\Xinner}{2}  
    
    \pgfmathsetmacro{\singleoffset}{0.225} 
    \pgfmathsetmacro{\intermediateoffset}{0.4} 
    \pgfmathsetmacro{\doubleoffset}{0.45} 
    \pgfmathsetmacro{\tripleoffset}{0.675} 
    \pgfmathsetmacro{\quadrupleoffset}{.9} 
    \pgfmathsetmacro{\quintupleoffset}{1.125} 
    
    \node[vertex] (v1) at (-6, \Ysep) {};
    \node[vertex] (v2) at (-4, \Ysep) {};
    \node[vertex] (v3) at (-\Xinner, \Ysep) {}; 
    
    \node[vertex] (v8) at (\Xinner, \Ysep) {}; 
    \node[vertex] (v9) at (4, \Ysep) {};
    \node[vertex] (v10) at (6, \Ysep) {};
    
    \node[vertex] (v4) at (-6, -\Ysep) {};
    \node[vertex] (v5) at (-4, -\Ysep) {};
    \node[vertex] (v6) at (-\Xinner, -\Ysep) {}; 
    
    \node[vertex] (v11) at (\Xinner, -\Ysep) {}; 
    \node[vertex] (v12) at (4, -\Ysep) {};
    \node[vertex] (v13) at (6, -\Ysep) {}; 
    
    \node[vertex] (v7) at (0, 0) {};
    
    \draw[pathstyle] (v1) -- (v2) -- (v3);
    \draw[pathstyle] (v4) -- (v5)-- (v6);
    \draw[pathstyle] (v8) -- (v9)-- (v10);
    \draw[pathstyle] (v11) -- (v12)-- (v13);

    \draw[pathstyle] (v3) -- (v7) -- (v11);
    \draw[pathstyle] (v6) -- (v7) -- (v8);
    
    \draw[rwthorangepathstyle]
    ($ (v3) + (0, \singleoffset) $)
    -- ($ (v7) + (0, \singleoffset) $) 
    -- ($ (v8) + (0, \singleoffset) $)
    -- ($ (v9) + (0, \singleoffset) $);
    \node[above right =.1cm and -.1cm of v9, color = rwthorange, font = \footnotesize]{$\ub_1 = 3$};

    \draw[rwthbluepathstyle]
    ($ (v1) + (0, -\singleoffset) $)
    -- ($ (v2) + (0, -\singleoffset) $)
    -- ($ (v3) + (-0, -\singleoffset) $) 
    -- ($ (v7) + (-.4, 0) $)
    -- ($ (v6) + (0, \singleoffset) $)
    -- ($ (v5) + (0, \singleoffset) $);
    \node[below left =-.1cm and -.1cm of v1, color = rwthblue, font = \footnotesize]{$\ub_2 = 1$};
    
    \draw[rwthlightbluepathstyle]
    ($ (v2) + (0, - \doubleoffset) $)
    -- ($ (v3) + (0, - \doubleoffset) $)
    -- ($ (v7) + (-.8, 0) $);
    \node[below left=.2cm and -.1cm of v2, color = rwthlightblue, font = \footnotesize]{$\ub_3 = 1$};
    
    \draw[rwthmagentapathstyle]
    ($ (v4) + (0, -\singleoffset) $)
    -- ($ (v5) + (0, -\singleoffset) $)
    -- ($ (v6) + (0, -\singleoffset) $)
    -- ($ (v7) + (0, -\singleoffset) $)
    -- ($ (v8) + (0, -\singleoffset) $)
    -- ($ (v9) + (0, -\singleoffset) $)
    -- ($ (v10) + (0, -\singleoffset) $);
    \node[below left =-.1cm and -.1cm of v4, color = rwthmagenta, font = \footnotesize] (label5){$\ub_4 = 3$};
    
    \draw[rwthbluepathstyle]
    ($ (v10) + (0, -\doubleoffset) $)
    -- ($ (v9) + (0, -\doubleoffset) $)
    -- ($ (v8) + (-0, -\doubleoffset) $) 
    -- ($ (v7) + (.625, -.1) $)
    -- ($ (v11) + (0, \singleoffset) $)
    -- ($ (v12) + (0, \singleoffset) $)
    -- ($ (v13) + (0, \singleoffset) $);
    \node[above right =-.1cm and -.1cm of v13, color = rwthblue, font = \footnotesize]{$\ub_5 = 1$};

    \draw[rwthsilver, ultra thick] ($ (v3) + (1, .25) $) coordinate (aux1) to[bend right=-30] ($ (v3) + (.3, -.9) $); 

    \draw[rwthsilver, ultra thick] ($ (v8) + (-1, .25) $) coordinate (aux2) to[bend right=30] ($ (v8) + (-.3, -.9) $); 

    \draw[rwthsilver, ultra thick] ($ (v9) + (-.75, .7) $) coordinate (aux3) to[bend right=30] ($ (v9) + (-.75, -.7) $); 
    
    \draw[rwthsilver, ultra thick] ($ (v4) + (.9, .7) $) coordinate (aux4) to[bend right=-30] ($ (v4) + (.9, -.7) $);

\end{tikzpicture}
		\caption{An instance of \ac{FZA} with five commodities, each with $\weight_i = 1$. The edges cut by the curved lines indicate an optimal cut set for the pricing function $f(x) = x$.}
		\label{fig:FZA-example}
	\end{figure}
	
	\paragraph{Connections to network pricing.} $\ac{FZA}$ is closely related to standard pricing problems, most notably \textsc{Tollbooth} on trees (see, e.g.,~\cite{Guruswami2005}). In \textsc{Tollbooth}, one is given a tree $T = (V,E)$ and a set of single-minded customers $i \in [k]$, each interested in buying a fixed path $P_i$ as long as the cost of $\pathGraph_i$ does not exceed their budget $b_i \geq 0$. The objective is to assign prices $p\colon E \to \mathbb{R}_+$ to edges in order to maximize the total revenue obtained by all customers, that is $\sum_{i \text{ served}} \sum_{e \in \pathGraph_i} p_e$. Here, we say that customer $i$ is served if $\sum_{e\in P_i} p_e \leq b_i$.
	
	\textsc{Tollbooth} and \textsc{FZA} differ in two key aspects: First, \textsc{Tollbooth} allows the network designer to set arbitrary prices on the edges, while in \textsc{FZA} the network designer is restricted to a binary decision for every edge $e$. This difference is quite significant as previous algorithms for \textsc{Tollbooth} heavily exploited the possibility of assigning high prices to single edges \cite{GamzuSegev2017loglog, TurkoByrka2024loglogCactus}. Second, in \textsc{Tollbooth} the revenue obtained from every served customer $i$ is a linear function, whereas \textsc{FZA} considers more general subadditive functions. Although the structural similarities between the problems allow for the transfer of certain algorithmic techniques, the aforementioned differences require substantially new ideas for the analysis.
	
	\subsection{Related Work}
	\label{subsec:RelatedWork}
	As outlined above, FZA is closely related to the \textsc{Tollbooth} problem initially studied by Guruswami et al.~\cite{Guruswami2005}. By assigning a single price to all edges, they show a $\mathcal{O}(\log k + \log m)$-approximation, where $k$ is the number of (unit-weight) commodities and $m$ the number of edges. Optimizing the choice of a single price is known to yield logarithmic approximations for network and item pricing problems in very large generality~\cite{Balcan2008SinglePrice,BohnleinKS21,Briest2008,CardinalMST2007,RochSM05}. Variants with limited supply have also been analyzed, e.g., in the context of posted pricing and assortment problems~\cite{Swamy2008}. However, these variants with limited supply do not fit particularly well with the context of fare design in public transport.
	
	Guruswami et al.~\cite{Guruswami2005} also describe an algorithm for solving rooted instances in polynomial time. Gamzu and Segev~\cite{GamzuSegev2017loglog} combine this algorithm with a clever recursive tree decomposition and buyer classification enabling enumeration in polynomial time, resulting in a $\mathcal{O}(\log m/\log \log m)$-approximation for \textsc{Tollbooth}.
	Recently, the latter algorithm was extended beyond trees to cactus graphs~\cite{TurkoByrka2024loglogCactus}, thereby relaxing single-mindedness (i.e., commodities are allowed to choose a path between their source and target). 
	Extending \ac{FZA} to cyclic graphs (such as cactus graphs) is of theoretical interest; however, edge cuts no longer define valid disjoint zones. If relaxed to merely counting the minimum number of cuts on any path between the endpoints of some commodity, applying the cactus-decomposition from \cite{TurkoByrka2024loglogCactus} is promising. Nonetheless, adapting this algorithm to \ac{FZA} poses non-trivial challenges, largely because it is impossible to accumulate a significant fraction of a commodity's willingness to pay on a single edge.
	A large body of research has focused on the \textsc{Highway} problem, i.e., \textsc{Tollbooth} restricted to a path. \textsc{Highway} is known to be strongly \textsf{NP}-hard~\cite{Elbassioni2012journal} while admitting a PTAS~\cite{GrandoniRothvossPTAS}. On trees, parameterized algorithms have been a popular research direction. Natural parameters are maximum path length $|\pathGraph|_{\max} \coloneqq \max\{|\pathGraph_i| \mid i\in [k]\}$, maximum budget $u_{\max} \coloneqq \max\{u_i\mid i\in [k]\}$, and congestion $\con \coloneqq \max_{e\in E}|\{i\in [k]\mid e\in \pathGraph_i\}|$. 
	Briest and Krysta~\cite{BriestKrystaSparsity2006} show that \textsc{Tollbooth} remains \textsf{APX}-hard even with constant congestion and bounded maximum path length.  
	They complement this result with efficient $\mathcal{O}(\log |\pathGraph|_{\max} + \log \con )$- or $\mathcal{O}(\log^2 |\pathGraph|_{\max})$-approximation algorithms.
	Similar parameters appear in the running times of our exact algorithms for \ac{FZA} on paths (cf. Appendix~\ref{app:parameterized}).
	
	From a practical perspective, fare design in public transport is inherently multifaceted, with many concerns---social equity, stakeholder trade-offs, price elasticity, and user accep\-tance---often outweighing purely algorithmic criteria.
	Studies of optimal public transport fare design generally follow two distinct objective paradigms: minimizing the deviation from reference fares \cite{Babel2003, HamacherSchoebel1995, HamacherSchoebel2004, SchieweSchoebel2025ATMOS, Schoebel2025}, e.g.,\ to ensure passenger fairness, and maximizing operator revenue based on passenger willingness to pay \cite{TariffZonesMueller22, Otto2017}. Recently, these paradigms have also been bridged by bi-objective models combining revenue- and passenger maximization \cite{SchieweSchoebel2024}.
	
	Most similar to our work are the models studied by Otto and Boysen \cite{Otto2017} and Schöbel and Urban \cite{Schoebel2025}. Both of these works consider general graphs and introduce mixed-integer programming (MIP) approaches to jointly optimize zoning and pricing problems for different zone topologies (connected, ring, and arbitrary zones) and pricing schemes (including counting-zones pricing). However, computational studies conducted in these works show that state-of-the-art solvers struggle to compute optimal solutions for small networks with 20 to 30 stations within a reasonable time.
	Regardless of the chosen objective, both papers identify the zoning subproblem as the primary computational bottleneck, proving it to be strongly \textsf{NP}-hard even on highly restrictive topologies like star graphs and trees. We remark that these reductions require that a zoning must be found w.r.t.\ a maximum number of zones which can be chosen in the reduction. 
	We provide alternative constructions (without a bound on the number of zones) proving that \ac{FZA} is \textsf{APX}-hard on star graphs and remains \textsf{NP}-hard even when restricted to paths.
	
	These practical and theoretical intractability results motivate our decision to isolate the zoning problem (similar to \cite{TariffZonesMueller22}) and underscore the necessity of approximation algorithms.
	
	\subsection{Our Contribution}
	\paragraph{Rooted instances.}
	In \cref{sec:singlesource}, we design an efficient algorithm that solves rooted instances of \ac{FZA} to optimality based on dynamic programming. It serves as a building block for the sublogarithmic approximation algorithm described in Section \ref{sec:SublogApprox}. Starting with the leaves of $T$, we iteratively compute the maximum profit that can be achieved when a fixed number of cuts lie on the path from the current vertex to the root. 
	
	\paragraph{Approximation results.} 
	In Sections~\ref{sec:SingleDensity} and \ref{sec:SublogApprox}, we present appro\-xi\-ma\-tion algorithms for \ac{FZA} based on a commonly used framework which we call \emph{Divide and Select approach}. It works by partitioning the set of commodities $[k]$ into $\ell$ classes, approximating an optimal solution restricted to the commodities of each class, and selecting the best result. 
	
	In Section~\ref{sec:SingleDensity}, we introduce the \textit{Single Density} algorithm, a randomized $\mathcal{O}(\log n)$-approx-imation for \ac{FZA} inspired by ``single price'' mechanisms in network and item pricing literature. The Single Density algorithm groups commodities into $\mathcal{O}(\log n)$ classes based on their \textit{density} (budget-to-length ratio). For each class $M_j$, the algorithm evaluates a cut set that exhausts a constant fraction of the upper bound $u_i$ of every commodity $i \in M_j$. We then randomly delete cuts from this solution to ensure that every $i \in M_j$ is served with constant probability.
	
	In Section~\ref{sec:SublogApprox}, we give a $\mathcal{O}(\log n / \log \log n)$-appro\-xi\-ma\-tion algorithm that partitions commodities based on a recursive tree decomposition. It improves the approximation ratio over the Single Density approach asymptotically, but is substantially more involved and also hides larger constants in the $\mathcal{O}$-notation. The algorithm builds upon a result for \textsc{Tollbooth} by Gamzu and Segev~\cite{GamzuSegev2017loglog}, but  necessitates modifications in several key aspects to translate subroutines that rely on assigning arbitrarily high prices to single edges to the edge-cutting problem \ac{FZA}.
	Lastly, in Appendix~\ref{app:ptas} we sketch how the PTAS for \textsc{Highway}~\cite{GrandoniRothvossPTAS} can be adapted to a PTAS for \ac{FZA} when $T$ is a path.
	
	\paragraph{Hardness results.}
	In Appendix~\ref{app:Hardness}, we reduce a variant of \textsc{Max 2-Sat} to \ac{FZA} implying that \ac{FZA} remains strongly \textsf{NP}-hard, even when restricted to paths. This result is tight in the sense that there exists a complementary PTAS (cf. Appendix~\ref{app:ptas}). For star graphs, we show \textsf{APX}-hardness using a similar reduction.
	
	\paragraph{Parameterized algorithms.} 
	Finally, in Appendix~\ref{app:parameterized}, we present exact parameterized algorithms for \ac{FZA} when $T$ is a path. In particular, we show that \ac{FZA} is in \textsf{XP} when parameterized with $\ub_{\max}$ or $\con$, and is in \textsf{FPT} when parameterized with $|\pathGraph|_{\max}$ or $\max\{\ub_{\max}, \con\}$ (see \cref{subsec:RelatedWork} for parameter definitions), all using dynamic programming. 
	Note that the instance which we create to show \textsf{APX}-hardness on star graphs is constant in all these parameters and thus rules out the possibility of similar results beyond paths.
	
	
	\section{Warm-Up---Solving Rooted Instances}
	\label{sec:singlesource}
	As a warm-up to familiarize the reader with the problem and our notation, we first describe and analyze a simple dynamic programming algorithm (DP) for solving \ac{FZA} efficiently on rooted instances. It is conceptually similar to the algorithm for solving rooted instances of \textsc{Tollbooth} described in \cite{Guruswami2005}. 
	The $\mathcal{O}(\frac{\log n}{\log \log n})$-approximation algorithm described in Section~\ref{sec:SublogApprox} uses this DP for \ac{FZA}, as well as a variant for \emph{generalized FZA on a path} (see Appendix~\ref{app:GeneralizedRooted}), as subroutines.
	
	\begin{proposition}
		\label{thm:SingleSource}
		\ac{FZA} on rooted instances can be solved in $\mathcal{O}(n^2k)$ time.
	\end{proposition}
	\begin{proof}
		Let $r$ denote the common endpoint (root) shared by all commodities $i \in [k]$, and let $t_i$ denote the other endpoint associated with commodity $i$.
		For each vertex $v\in V$, let $Q_v\subseteq E$ be the (unique) $r$--$v$-path in $T$. Let $T[v]$ be the subtree of $T$ induced by all vertices $w\in V$ with $v\in Q_w$. We call vertex $w$ a \emph{child of $v$} if $Q_w = Q_v\cup\{\{v,w\}\}$. Let $\mathcal{C}(v)$ be the set of $v$'s children.
		For each vertex $v\in V$ and each number $x\in \{0,1, \ldots, |E|\}$, let $R_v(x)$ denote the maximal revenue that can be obtained solely from commodities $i\in [k]$ with
		$v\in \pathGraph_i$, given that we cut exactly $x$ edges of $Q_v$. If $v$ is a leaf,
		\[
		R_v(x)= \sum_{i\in [k]}\bigl\{\weight_i \cdot f(x) \,\big\vert\, \sink_i = v,\ x\le \ub_i \bigr\}.
		\]
		For non-leaf vertices $v\in V,$ we compute $R_v(x)$ in a bottom-up manner. Starting at the leaves, we traverse $T$ towards the root and apply the following recursion.
		\begin{equation}
			\label{eq:RootedDP_Recursion}
			R_v(x)= \sum_{i\in [k]}\bigl\{\weight_i \cdot f(x) \,\big\vert\, \sink_i = v,\ x\le \ub_i \bigr\}
			+\sum_{w\in \mathcal{C}(v)}\max\bigl\{R_w(x),\ R_w(x+1) \bigr\}.
		\end{equation}
		
		In addition, we construct cut sets $\cutset_{v,x}\subseteq E$ associated to the values $R_v(x)$. Each $\cutset_{v,x}$ consists solely of edges from $T[v]$ such that the union of $\cutset_{v,x}$ with any cut set $\cutset\subseteq Q_v$ of size $x$, yields a revenue of $R_v(x)$. Using \cref{eq:RootedDP_Recursion}, we compute $\cutset_{v,x}$ as follows. 
		Partition $\mathcal{C}(v)$ into the sets $\mathcal{C}_1(v)$ and $\mathcal{C}_2(v)\coloneqq\mathcal{C}(v)\setminus{\mathcal{C}_1(v)}$ consisting of all children $w$ for which the maximum in \cref{eq:RootedDP_Recursion} is attained in the first and the second term, respectively.
		Then,
		\[
		\cutset_{v,x}\coloneqq\bigcup_{w\in \mathcal{C}_1(v)} \cutset_{w,x} \ \cup \ \bigcup_{w\in \mathcal{C}_2(v)} \bigl( \cutset_{w,x+1} \cup \{\{v,w\}\}\bigr).
		\]
		
		The procedure terminates upon arrival at the root. $R_r(0)$ corresponds to the optimal objective value and $\cutset^* \coloneqq \cutset_{r,0}$ is a corresponding optimal solution.
	\end{proof}
	
	
	\section{Single Density Approximation Algorithm}
	\label{sec:SingleDensity}
	Our approximation algorithms rely on a common framework which we call \textit{Divide and Select} approach: Given an instance of \ac{FZA}, divide the set of commodities into $\ell$ disjoint classes $M_1, \ldots, M_{\ell}$, and compute a cut set $F_j\subseteq E$ for each $j\in [\ell]$ satisfying
	\begin{equation*}
		\label{eq:divide_and_select1}
		\mathbb{E}[\rev{M_j}(F_j)] \geq \alpha\cdot \rev{M_j}(F^*),
	\end{equation*}
	where $\cutset^*$ is an optimal solution and $\rev{M_j}(\cutset_j) = \sum_{i\in M_j}\rev{i}(\cutset_j)$.
	After computing a $F_j$ for every $j \in [\ell]$, we select some  $\cutset_{\ALG} \in \argmax\left\{\rev{M_j}(\cutset_j) \mid j\in [\ell]\right\}$.
	This directly yields a performance guarantee of $\frac{\alpha}{\ell}$ in expectation.
	
	\paragraph{Description of the Single Density algorithm.}
	Following the Divide and Select approach, the \textit{Single Density algorithm} first partitions the set of commodities $[k]$ into $\ell = \lceil\log_2(n)\rceil + 1$ classes $M_0, \dots, M_{\ell-1}$, depending on their budget-to-length ratio
	$d_i \coloneqq \frac{u_i}{|\pathGraph_i|} \in [0,1]$ for $i\in [k]$. We call $d_i$ the \textit{density} of commodity $i$. More specifically, we construct the sets
	\begin{equation*}
		M_j \coloneqq \bigl\{i \in [k] \ \big\vert \ \ub_i \geq 1 \text{ and } d_i \in (2^{-j}, 2^{1-j}]\bigr\}, 
	\end{equation*}
	for each $j \in \{1, \dots, \lceil\log_2(n)\rceil\}$, and a class $M_0 \coloneqq \{i \in [k] \ \big\vert \ \ub_i =0\}$ covering the remaining commodities with $d_i = 0$.
	Fix $\cutset_0 = \emptyset$. We construct the cut sets $\cutset_j$ for $j\in \{1, \dots, \lceil\log_2(n)\rceil\}$ as follows.
	Designate an arbitrary vertex $r\in V$ as the root of $T$. For an edge $e$ of $T$, we denote by $\text{dist}(r,e)$ the number of edges other than $e$ on the unique path from $e$ to the root $r$.
	We then construct the following candidate solutions for each $j$ and every ${\eStart} \in \{0, \dots, 2^{j+1}-1\}$
	\begin{equation*}
		\cutset_{j, \eStart} \coloneqq \bigl\{e\in E \ \big\vert \ \text{dist}(r,e) \equiv \eStart \mod \ 2^{j+1}\}.
	\end{equation*}
	\cref{fig:SingleDensity} illustrates this definition. Note that directly evaluating $\rev{}(\cutset_{j, \eStart})$ for every ${\eStart}$ and choosing the best is not sufficient. Since $T$ is a tree, $\cutset_{j, \eStart}$ may include two cuts on a commodity $i$ with $\ub_i = 1$, leading to $i$ dropping out. This may happen if path $P_i$ is not entirely contained in a simple path from the root to one of the leaves of $T$. Commodity 3 in \cref{fig:SingleDensity} illustrates this problem.
	To address this, we alter every set $\cutset_{j, \eStart}$ by iterating through all cuts $e \in \cutset_{j,\eStart}$ and removing $e$ with probability $1/2$. We denote the resulting cut sets as $\cutset_{j, \eStart}^{\text{(rand)}}$, choose a revenue-maximizing solution $\cutset_j$ among all the $\cutset_{j, \eStart}^{\text{(rand)}}$, and finally return $\cutset_{\ALG} \in \argmax\left\{\rev{M_j}(\cutset_j) \mid j\in [\ell]\right\}.$
	
	\begin{figure}[tbph]
		\centering
		\begin{tikzpicture}[
        vertex/.style = {fill = black, minimum size = .2cm, shape = circle},
        graphEdge/.style = {draw = black},
        scale = .85
    ]
    \pgfmathsetmacro{\singleoffset}{0.3}
    \pgfmathsetmacro{\doubleoffset}{0.45}
    \pgfmathsetmacro{\tripleoffset}{.6}
    \pgfmathsetmacro{\quadrupleoffset}{.75}
    \pgfmathsetmacro{\quintupleoffset}{.9}
    \pgfmathsetmacro{\sextupleoffset}{1.05}
    \node[vertex] (r) at (0,0) {};
    \node[vertex] (1c) at (1.5,0) {};
    \node[vertex] (2c) at (3,0) {};
    \node[vertex] (3t) at (4.5,.8) {};
    \node[vertex] (3b) at (4.5,-.8) {};
    \node[vertex] (4t) at (6,.8) {};
    \node[vertex] (4b) at (6,-.8) {};
    \node[vertex] (5bt) at (7.5,-.2) {};
    \node[vertex] (5bb) at (7.5,-1.4) {};
    \node[vertex] (6bt) at (9,-.2) {};
    \node[vertex] (6bb) at (9,-1.4) {};
    \node[vertex] (7bt) at (10.5,-.2) {};
    \node[vertex] (7bb) at (10.5,-1.4) {};
    \node[vertex] (8bb) at (12,-1.4) {};

    \begin{pgfonlayer}{bg}    
        \path[graphEdge] (r) -- (1c) -- (2c) -- (3t) -- (4t);
        \path[graphEdge] (2c) -- (3b) -- (4b);
        \path[graphEdge] (7bt) -- (6bt) -- (5bt) -- (4b) -- (5bb) -- (6bb) -- (7bb) -- (8bb);
    \end{pgfonlayer}
    \node[draw = none, fill = white, shape = rectangle] (I1) at (.75, .6) {\small \textcolor{rwthmagenta}{$\ub_1 = 3$}};
    \node[draw = none, fill = white, shape = rectangle] (I2) at (11.25, -.8) {\small \textcolor{rwthmagenta}{$\ub_2 = 1$}};
    \node[draw = none, fill = white, shape = rectangle] (I3) at (4.5, 0) {\small \textcolor{rwthblue}{$\ub_3 = 1$}};
    \node[draw = none, fill = white, shape = rectangle] (I4) at (8.25, .4) {\small \textcolor{rwthblue}{$\ub_4 = 1$}};
    \node[draw = none, fill = white, shape = rectangle] (I5) at (9, -.8) {\small \textcolor{rwthblue}{$\ub_5 = 2$}};

    \node[draw = none] (lable_r) at (-0.4, 0) {\small{$r$}};
    \draw[-, densely dashed, very thick, color = rwthmagenta]
    ($ (r) + (0, \singleoffset) $)
        -- ($ (1c) + (0, \singleoffset) $)
        -- ($ (2c) + (0, \singleoffset) $)
        -- ($ (3t) + (0, \singleoffset) $)
        -- ($ (4t) + (0, \singleoffset) $);
    \draw[-, densely dashed, very thick, color = rwthmagenta]
    ($ (7bb) + (.075, \singleoffset) $)
        -- ($ (8bb) + (0, \singleoffset) $);
    \draw[-, densely dotted, very thick, color = rwthblue]
    ($ (3t) + (0, -\singleoffset) $)
        -- ($ (2c) + (\doubleoffset,0) $)
        -- ($ (3b) + (0, \singleoffset) $);
    \draw[-, line width = 1mm, color = white]
    ($ (6bt) + (0, \singleoffset) $)
        -- ($ (5bt) + (0, \singleoffset) $)
        -- ($ (4b) + (-\doubleoffset,0) $)
        -- ($ (5bb) + (0, -\singleoffset) $);
    \draw[-, densely dotted, very thick, color = rwthblue]
    ($ (6bt) + (0, \singleoffset) $)
        -- ($ (5bt) + (0, \singleoffset) $)
        -- ($ (4b) + (-\doubleoffset,0) $)
        -- ($ (5bb) + (0, -\singleoffset) $);
    \draw[-, densely dotted, very thick, color = rwthblue]
    ($ (7bt) + (-.075, -\singleoffset) $)
        -- ($ (6bt) + (0, -\singleoffset) $)
        -- ($ (5bt) + (0, -\singleoffset) $)
        -- ($ (4b) + (\doubleoffset,0) $)
        -- ($ (5bb) + (0, \singleoffset) $)
        -- ($ (6bb) + (0, \singleoffset) $)
        -- ($ (7bb)  + (-.075, \singleoffset) $);

    \node[vertex] at (6,-.8) {};
    \begin{pgfonlayer}{bg}    
        \draw[line width = 1mm, color = rwthmagenta] (3t) -- (4t);
        \draw[line width = 1mm, color = rwthmagenta] (3b) -- (4b);
        \draw[line width = 1mm, color = rwthmagenta] (7bb) -- (8bb);
    \end{pgfonlayer}
    \begin{pgfonlayer}{bg}    
        \draw[line width = 1mm, color = rwthblue] (5bt) -- (6bt);
        \draw[line width = 1mm, color = rwthblue] (5bb) -- (6bb);
    \end{pgfonlayer}
\end{tikzpicture}
		\caption{An illustration of the sets $M_j$ and the construction of the candidate solutions $F_{j, \theta}$. The dashed commodities colored in magenta form the set $M_1$, and the blue, dotted commodities form the set $M_2$. Bold edges in different colors indicate the solutions $F_{1,3}$ and $F_{2,5}$ which are optimal for $M_1$ and $M_2$, respectively. Note that there is no $\eStart \in \{0,\dots, 7\}$ such that $\rev{3}(F_{j,\eStart}) > 0$.}
		\label{fig:SingleDensity}
	\end{figure}
	
	\begin{restatable}{theorem}{SingleDensity}
		\label{thm:SingleDensity}
		The Single Density algorithm is an $\mathcal{O}(\log n)$-approximation for \ac{FZA} in expectation.  
	\end{restatable}
	
	\begin{proof}[Proof (Sketch)]
		We only sketch the analysis and defer the full proof of \cref{thm:SingleDensity} to Appendix~\ref{app:SingleDensityProofs}.
		It follows from the probabilistic method that, after the randomization step, for every $j > 0$, there exists some ${\eStart} \in \{0, \dots, 2^{j+1}-1\}$ such that the solution $\cutset_{j,\eStart}^{\text{(rand)}}$ fulfills 
		\begin{equation}
			\label{eq:single_density_expectation}
			\mathbb{E}[\rev{M_j}(\cutset_{j,\eStart}^{\text{(rand)}})] \geq \frac{1}{24} \rev{M_j}(\cutset^*),
		\end{equation}
		which implies the desired statement. To prove Equation (\ref{eq:single_density_expectation}), we first show for all $i \in M_j$ that $\mathbb{E}\bigl[\big|P_i \cap \cutset_{j, \eStart}^{(\text{rand})}\big|\bigr] \geq \tfrac{1}{8}|P_i \cap F^*|$. 
		Note that for non-linear pricing functions $f$, the desired inequality $\mathbb{E}[\rev{M_j}(\cutset_{j, \eStart}^{(\text{rand})})] \geq \alpha \cdot \rev{M_j}(P_i \cap F^*)$ does not follow immediately. However, using the subadditivity of $f$, we show that bounding the expected number of cuts on the path of each commodity $i$ translates directly into a lower bound on the revenue that $\cutset_{j, \eStart}^{(\text{rand})}$ extracts from $i$. 
		In particular, we show that $\rev{i}(\cutset_{j, \eStart}^{(\text{rand})})$ is a constant fraction of $\rev{i}(F^*)$.
	\end{proof}
	
	In Appendix~\ref{app:SingleDensityProofs}, we also present a deterministic $\mathcal{O}(\log n)$-approximation  hiding smaller constants in the $\mathcal{O}$-notation for instances where $f(0)>0$ or $T$ is a path. Additionally, we introduce a simpler and arguably more intuitive algorithm which in turn relies more on randomization. Further, its analysis results in larger constants in the approximation guarantee.
	
	
	\section{Sublogarithmic Approximation}
	\label{sec:SublogApprox}
	
	In this section, we present an $\mathcal{O}({\log n}/{\log \log n})$-approximation algorithm for \ac{FZA}. The algorithm and its analysis build on the work of Gamzu and Segev~\cite{GamzuSegev2017loglog} for \textsc{Tollbooth}.
	We apply the Divide and Select approach and partition the set of commodities into $\ell \in \mathcal{O}(\log n / \log \log n)$ classes $M_j$, $j\in [\ell]$, which we obtain by iteratively computing \emph{almost balanced $d$-decompositions} (defined below) of the input graph. 
	Such decompositions always exist and can be computed in polynomial time due to Frederickson and Johnson~\cite{FredericksonJohnson1980almostBalanced}.
	
	\begin{definition}[see \cite{FredericksonJohnson1980almostBalanced}]
		\label{def:almost_balanced}
		Let $T = (V,E)$ be a tree and $d \in \mathbb{N}$. An \emph{almost balanced $d$-decomposition} of $T$ is a partition of the edges of $T$ into $d$ edge-disjoint subtrees $T_1, \ldots, T_d$ such that each of these subtrees contains between $|E|/(3d)$ and $3|E|/d$ edges.
	\end{definition}
	
	Starting with the input tree $T$, we iteratively refine the current decomposition by computing an almost balanced $d$-decomposition of each component. 
	Given an instance  of \ac{FZA} with  tree  decomposition $\mathcal{T} = \{T_1, \ldots, T_d\}$, we say that commodity $i$ gets \emph{separated} by $\mathcal{T}$ if $P_i$ intersects with at least two subtrees in $\mathcal{T}$. We say that an instance $\mathcal{I}$ of \ac{FZA} is \emph{separated by a tree decomposition $\mathcal{T} = \{T_1, \ldots, T_d\}$} if every commodity in $\mathcal{I}$ is separated by $\mathcal{T}$, and $d \in \mathcal{O}(\log(n)^{1/2})$.
	As described below, class $M_j$ contains all those commodities $i\in [k]$,  which are separated for the first time in the $(j+1)$-th iteration of this decomposition process. As it turns out (see proof of Corollary \ref{cor.main}),
	to obtain a constant-factor approximation for each class $M_j$, it suffices to solve \ac{FZA} on instances in which each commodity is separated by a given tree decomposition.  
	
	\begin{theorem}
		\label{lem:Trees_constant}
		There exists a constant-factor approximation algorithm (in expectation) for solving \ac{FZA} on instances $\mathcal{I}$ with a given tree decomposition $\mathcal{T}$ that separates $\mathcal{I}$.
	\end{theorem}
	
	Before we prove Theorem \ref{lem:Trees_constant}, we show how our main result follows as a consequence.
	
	\begin{corollary}\label{cor.main}
		There is an efficient algorithm to compute an $\mathcal{O}\left(\frac{\log n}{\log \log n}\right)$-appro\-xi\-ma\-tion for \ac{FZA} in expectation.
	\end{corollary}
	
	\begin{proof} [Proof (of \cref{cor.main}).]  The proof goes along the same lines as in~\cite{GamzuSegev2017loglog}.  For $j\in \{1,2, \ldots\}$, we iteratively construct an 
		almost balanced $d$-decomposition $\mathcal{T}^{(j)}$ of $E$ into subtrees $T_1, \dots, T_d$, with each $\mathcal{T}^{(j+1)}$ being a refinement of decomposition $\mathcal{T}^{(j)}$. We initialize $\mathcal{T}^{(1)} = \{E\}$, and set $d = \lceil (\log_2 n)^{1/2} \rceil$ for $j > 1$. Given $\mathcal{T}^{(j)}$, we now compute an almost balanced $d$-decomposition of each subtree in $\mathcal{T}^{(j)}$, or---if some subtree in $\mathcal{T}^{(j)}$ has at most $d$ edges---decompose this subtree into its individual edges. By taking the union of all these decompositions, we obtain $\mathcal{T}^{(j+1)}$ as a refinement of $\mathcal{T}^{(j)}$.
		We repeat this process until at some level $\ell$ the decomposition $\mathcal{T}^{(\ell)}$ decomposes $T$ into its individual edges. Note that each subtree in $\mathcal{T}^{(j)}$ contains at most $(3/d)^{j-1}\cdot|E|$ edges and thus the number of decomposition levels is bounded by
		\[
		\ell \in \mathcal{O}(\log_d |E|) = \mathcal{O}\left(\frac{\log |E|}{\log \log |E|}\right).
		\]
		
		We now partition the commodities into classes $M_j$ as follows: Commodity $i \in [k]$ gets assigned to the unique level $j$ in which $i$ gets separated by $\mathcal{T}^{(j+1)}$ but not by $\mathcal{T}^{(j)}$.        
		Furthermore, we assign to each subtree $T' \in \mathcal{T}^{(j)}$ those commodities $i \in M_j$, whose path is entirely contained in $T'$. We observe that the revenue that some solution $F$ obtains from a commodity assigned to some subtree $T' \in \mathcal{T}^{(j)}$ only depends on $F \cap T'$.
		
		Let $F^*$ be an optimal solution, and let $j$ be the level of the decomposition for which $\rev{M_j}(F^*)$ is maximized. Since the subtrees $T'$ within a decomposition $\mathcal{T}^{(j)}$ are edge disjoint, we can treat them independently. For each $T' \in \mathcal{T}^{(j)}$, we consider the subinstance consisting of $T'$ and the set of commodities assigned to $T'$. Every commodity in this subinstance is separated by $\mathcal{T}^{(j+1)}$, and thus by \cref{lem:Trees_constant} we can approximate an optimal solution up to a constant factor. As the subtrees in $\mathcal{T}^{(j)}$ are mutually independent, this implies \cref{cor.main}.
	\end{proof}
	
	In the remainder of this section, we describe and analyze the algorithm used in the proof of \cref{lem:Trees_constant}.  As in \cite{GamzuSegev2017loglog}, we apply different subroutines for two cases that we call the \emph{non-skeleton case} and the \emph{skeleton case}, and show that the revenue obtained by the better of the two solutions  is a constant factor approximation. 
	While the non-skeleton case is similar to~\cite{GamzuSegev2017loglog}, the subroutine for the skeleton case and its analysis differ significantly.
	
	
	\paragraph{Border vertices and the skeleton.}
	Let $\mathcal{T}=\{T_1, \ldots, T_d\}$ be an almost balanced $d$-decomposition of a tree $T=(V,E)$. A vertex $v$ is a \textit{border vertex} if it is contained in at least two subtrees of the decomposition $\mathcal{T}$. We denote the set of border vertices by $V_B$. 
	The \textit{skeleton} $\mathcal{S} \subseteq T$ of $\mathcal{T}$ is the subtree spanned by all border vertices of $T$. 
	For convenience, we identify trees with their edge sets throughout the remainder of this section. 
	We call every non-border vertex $v \in \mathcal{S}$ of degree greater or equal to three \emph{junction vertex}, and denote the set of all junction vertices by $V_J$.
	Figure~\ref{fig:skeleton} illustrates these definitions.
	
	\begin{figure}[tbph]
		\centering
		    \resizebox{.72\textwidth}{!}{
    \begin{tikzpicture}
    \draw[-, gray, dashed] (-.5, 1.5) -- (1.75, 1.5) -- (1.75, -1.25) -- (-.5, -1.25);
    \draw[-, gray, dashed] (1.75, -.5) -- (1.75, -1.25) -- (4.75, -1.25) -- (4.75, 1.5) -- (1.75, 1.5);
    \draw[-, gray, dashed] (4.75, -1.25) -- (7.5, -1.25) -- (7.5, 1.5) -- (4.75, 1.5);
    \draw[-, gray, dashed] (10.5, -1.25) -- (7.5, -1.25) -- (7.5, -1);
    \draw[-, gray, dashed] (7.5, .5) -- (10.5, .5);
    \draw[-, gray, dashed] (7.5, 1.5) -- (7.5, 2.75);
    \draw[-, gray, dashed] (.75, 1.5) -- (.75, 2.75);
    \draw[-, gray, dashed] (4.75, 1.5) -- (4.75, 2.75);

    \node[draw = none] (T1) at (.75,-0.5) {\textcolor{gray}{$T_1$}};
    \node[draw = none] (T2) at (4.25,-0.5) {\textcolor{gray}{$T_2$}};
    \node[draw = none] (T3) at (6,-.5) {\textcolor{gray}{$T_3$}};
    \node[draw = none] (T4) at (4.25,1.875) {\textcolor{gray}{$T_4$}};
    \node[draw = none] (T5) at (10,1) {\textcolor{gray}{$T_5$}};
    \node[draw = none] (T6) at (10,-0.5) {\textcolor{gray}{$T_6$}};
    
    \node[draw = none, fill = white, shape = circle, inner sep = 0pt, outer sep = 0pt, minimum size = .08cm] (L1) at (1.75, 0) {$v_1$};
    \node[draw, fill = black, scale = 0.35, shape = circle] (A1) at (0,0) {};
    \node[draw, fill = black, scale = 0.35, shape = circle] (A2) at (1,0) {};
    \node[draw = none, fill = rwthmagenta, scale = 0.5, shape = circle] (A3) at (1.75,.5) {};
    \node[draw, fill = black, scale = 0.35, shape = circle] (A4) at (0,1) {};
    \node[draw, fill = black, scale = 0.35, shape = circle] (A5) at (1,1) {};

    \draw[-] (A1) -- (A2) -- (A3) -- (A5) -- (A4);

    \node[draw = none, fill = white, shape = circle, inner sep = 0pt, outer sep = 0pt, minimum size = .1cm] (L2) at (4.75, 0) {$v_3$};
    \node[draw = none, fill = white, shape = circle, inner sep = 0pt, outer sep = 0pt, minimum size = .1cm] (L5) at (2.25, 1.5) {$v_2$};
    \node[draw = rwthlightblue, line width = 0.2em, fill = none, scale = 0.45, shape = circle] (B1) at (2.75,.5) {};
    \node[draw, fill = black, scale = 0.35, shape = circle] (B2) at (2.25,-.5) {};
    \node[draw, fill = black, scale = 0.35, shape = circle] (B3) at (3.25,-.5) {};
    \node[draw, fill = black, scale = 0.35, shape = circle] (B4) at (3.75,.5) {};
    \node[draw = none, fill = rwthmagenta, scale = 0.5, shape = circle] (B5) at (4.75,.5) {};
    \node[draw = none, fill = rwthmagenta, scale = 0.5, shape = circle] (B6) at (2.75,1.5) {};

    \draw[-, line width = 2.5pt, rwthblue] (A3) -- (B1) -- (B4) -- (B5);
    \draw[-, line width = 2.5pt, rwthblue] (B1) -- (B6);
    \draw[-] (B3) -- (B1) -- (B2);

    \node[draw = rwthlightblue, line width = 0.2em, fill = none, scale = 0.45, shape = circle] (B1) at (2.75,.5) {};

    \node[draw = none, fill = white, shape = circle, inner sep = 0pt, outer sep = 0pt, minimum size = .1cm] (L3) at (7.5, 1.5) {$v_4$};
    \node[draw = none, fill = white, shape = circle, inner sep = 0pt, outer sep = 0pt, minimum size = .1cm] (L4) at (7.5, -.5) {$v_5$};
    \node[draw = rwthlightblue, line width = 0.2em, fill = none, scale = 0.45, shape = circle] (C1) at (5.75,.5) {};
    \node[draw, fill = black, scale = 0.35, shape = circle] (C2) at (6.5,1) {};
    \node[draw = none, fill = rwthmagenta, scale = 0.5, shape = circle] (C3) at (7.5,1) {};
    \node[draw, fill = black, scale = 0.35, shape = circle] (C4) at (6.5,0) {};
    \node[draw = none, fill = rwthmagenta, scale = 0.5, shape = circle] (C5) at (7.5,0) {};

    \draw[-, line width = 2.5pt, rwthblue] (B5) -- (C1) -- (C2) -- (C3);
    \draw[-, line width = 2.5pt, rwthblue] (C1) -- (C4) -- (C5);

    \node[draw = rwthlightblue, line width = 0.2em, fill = none, scale = 0.45, shape = circle] (C1) at (5.75,.5) {};

    \node[draw, fill = black, scale = 0.35, shape = circle] (D1) at (2.25,2.5) {};
    \node[draw, fill = black, scale = 0.35, shape = circle] (D2) at (1.25,2.5) {};
    \node[draw, fill = black, scale = 0.35, shape = circle] (D3) at (3.25,2.5) {};
    \node[draw, fill = black, scale = 0.35, shape = circle] (D4) at (4.25,2.5) {};

    \draw[-] (D2) -- (D1) -- (B6) -- (D3) -- (D4);

    \node[draw, fill = black, scale = 0.35, shape = circle] (E1) at (8.5,1) {};
    \node[draw, fill = black, scale = 0.35, shape = circle] (E2) at (8.5,2) {};
    \node[draw, fill = black, scale = 0.35, shape = circle] (E3) at (9.5,1) {};
    \node[draw, fill = black, scale = 0.35, shape = circle] (E4) at (9.5,2) {};

    \draw[-] (C3) -- (E1) -- (E2);
    \draw[-] (E3) -- (E1) -- (E4);

    \node[draw, fill = black, scale = 0.35, shape = circle] (F1) at (8.5,0) {};
    \node[draw, fill = black, scale = 0.35, shape = circle] (F2) at (8.5,-.75) {};
    \node[draw, fill = black, scale = 0.35, shape = circle] (F3) at (9.5,0) {};
    \node[draw, fill = black, scale = 0.35, shape = circle] (F4) at (9.5,-.75) {};

    \draw[-] (C5) -- (F1) -- (F3);
    \draw[-] (C5) -- (F2) -- (F4);

    \node[draw = none, fill = white, shape = circle, inner sep = 0pt, outer sep = 0pt, minimum size = .1cm] (I1) at (.5, .5) {\textcolor{rwthlightblue}{$P_1$}};
    \node[draw = none, fill = white, shape = circle, inner sep = 0pt, outer sep = 0pt, minimum size = .1cm] (I2) at (6.75, .5) {\textcolor{rwthlightblue}{$P_2$}};
    \node[draw = none, fill = white, shape = circle, inner sep = 0pt, outer sep = 0pt, minimum size = .1cm] (I3) at (3.35, 1.1) {\textcolor{rwthlightblue}{$P_3$}};
    \draw[-, densely dotted, very thick, color = rwthlightblue] (-.15,.25) -- (1,.25) -- (1.35, .5) -- (1, .75) -- (-.15, .75);
    \draw[-, densely dotted, very thick, color = rwthlightblue] (3.4, 2.25) -- (3, 1.5) -- (3, 0.75) -- (5.75,.75) -- (6.5, 1.25);
    \draw[-, densely dotted, very thick, color = rwthlightblue] (7.65, .75) -- (6.5, .75) -- (6.15, .5) -- (6.5, .25) -- (9.75, .25);

    \end{tikzpicture}
    }
    
		\caption{A decomposition $\mathcal{T} = \{T_1, \dots, T_6\}$ of some tree $T$. The larger, magenta vertices $v_1,\ldots,v_5$ are border vertices, and the hollow, light blue vertices are junction vertices. The skeleton (bold blue edges) is the subtree spanned by the border vertices. Commodity $1$ is not separated by $\mathcal{T}$ since $P_1 \subseteq T_1$, commodities $2$ and $3$ are separated by $\mathcal{T}$.}
		\label{fig:skeleton}
	\end{figure}
	
	To prove \cref{lem:Trees_constant}, we consider an instance of \ac{FZA} in which all commodities in $M=[k]$ are separated by some decomposition $\mathcal{T}=\{T_1, \ldots, T_d\}$ of the underlying tree $T=(V,E)$ with $d = \lceil (\log_2 n)^{1/2} \rceil$. Let $\mathcal{S}\subseteq E$ be the skeleton of $\mathcal{T}$ and let $E\setminus{\mathcal{S}}$ denote its non-skeleton edges. 
	We fix some optimal solution $F^*$. To compute a cut set $F_{\ALG}$ whose revenue---in expectation---diverges at most by a constant factor from $\rev{}(F^*)$, we run two different algorithms. The first one computes a cut set $F^{E\setminus\mathcal{S}}\subseteq E\setminus{\mathcal{S}}$ on the non-skeleton edges whose revenue in expectation is at least $\frac{1}{8}$ of the revenue that $F^*$ restricted to non-skeleton edges obtains. The second algorithm (described in Section~\ref{subsec:skeleton}) computes a cut set $F^{\mathcal{S}}\subseteq \mathcal{S}$ which solely consists of skeleton edges. We show in \cref{l.skeleton} that the expected revenue of $F^{\mathcal{S}}\subseteq \mathcal{S}$ is at least $\frac{1}{64}$ of the revenue which $F^*$ obtains when restricted to skeleton edges. The algorithm eventually returns the better solution $F_{\ALG}$ among $F^{E\setminus\mathcal{S}}$ and $F^{\mathcal{S}}$.
	
	\paragraph{Computing a cut set outside the skeleton.} 
	To construct a cut set $F^{E\setminus\mathcal{S}} \subseteq E\setminus \mathcal{S}$ consisting solely of non-skeleton edges, we follow the lines of \cite{GamzuSegev2017loglog}.  
	First, we contract the skeleton $\mathcal{S}$ to a single vertex $v_{\mathcal{S}}$. The new, contracted subtree decomposes into subtrees $T_1, \ldots, T_{d'}$, each rooted in $v_{\mathcal{S}}$.  We treat these subtrees independently and compute partial solutions $F_j$ for each subtree. Finally, we merge all these solutions and obtain $F^{E \setminus \mathcal{S}} = \bigcup_{j \in[d']} F_j$.
	
	For each $j\in [d']$, we call subtree $T_j$ \emph{active} with probability $1/2$, and \emph{inactive} otherwise. For inactive subtrees $T_j$, we set $F_j = \emptyset$. For every active subtree $T_j$, we consider a rooted instance $\mathcal{I}(T_j)$ on $T_j$ with root $v_{\mathcal{S}}$. A commodity $i \in [k]$ is present in $\mathcal{I}(T_j)$ if exactly one endpoint of $P_i$ is an inner vertex of $T_j$, and the other endpoint lies in $\mathcal{S}$ or some inactive other subtree.
	We now run the algorithm from \cref{sec:singlesource} for every constructed rooted instance to compute an optimal solution $\cutset_j$ for each $\mathcal{I}(T_j)$, and let $\cutset^{E \setminus \mathcal{S}} = \bigcup_{j \in [d']} F_j$. The following lemma states that in expectation $F^{E \setminus \mathcal{S}}$ extracts at least $\tfrac{1}{8}$ of the potential revenue from $E\setminus{\mathcal{S}}$. As the proof of \cref{l.non-skeleton} is very similar to \cite{GamzuSegev2017loglog}, we defer it to Appendix~\ref{app:non-skel}.
	
	\begin{restatable}{lemma}{nonskelApproximation}
		\label{l.non-skeleton}
		Let $\mathcal{I}$ be an instance of \ac{FZA} that is separated by some given tree decomposition $\mathcal{T}$ with associated skeleton $\mathcal{S}$. Further, let $F^*$ be an optimal solution to $\mathcal{I}$. There exists an efficient algorithm which computes a cut set $F^{E\setminus\mathcal{S}}\subseteq E\setminus{\mathcal{S}}$ such that
		\[
		\mathbb{E}[\rev{}(F^{E\setminus\mathcal{S}})] \geq \frac{1}{8}\rev{}(F^*\cap (E\setminus{\mathcal{S}})).
		\]
	\end{restatable}
	
	\paragraph{Computing a cut set inside the skeleton.}
	We defer the description of the algorithm for computing the cut set $F^{\mathcal{S}}\subseteq \mathcal{S}$ and its analysis to \cref{subsec:skeleton}, where we  prove:
	
	\begin{lemma} \label{l.skeleton} Let $\mathcal{I}$ be an instance of \ac{FZA} that is separated by some given tree decomposition $\mathcal{T}$ with associated skeleton $\mathcal{S}$. Further, let $F^*$ be an optimal solution to $\mathcal{I}$.
		There exists an efficient algorithm which computes a cut set $F^{\mathcal{S}}\subseteq \mathcal{S}$ such that
		\[
		\mathbb{E}[\rev{}(F^{\mathcal{S}})] \geq \frac{1}{64}\rev{}(F^*\cap \mathcal{S}).
		\]
	\end{lemma}
	
	Combining Lemmata \ref{l.non-skeleton} and \ref{l.skeleton}, we obtain \cref{lem:Trees_constant} as an immediate corollary.
	
	\begin{proof}[Proof (of \cref{lem:Trees_constant}).]
		Let $F^*$ be an optimal solution. By the subadditivity of the pricing function $f$, we have $\rev{i}(F^* \cap \mathcal{S}) + \rev{i}(F^* \cap (E \setminus \mathcal{S})) \geq \rev{i}(F^*)$ for every commodity $i\in[k]$. Hence, summing over all commodities yields that at least one of $\rev{}(F^* \cap \mathcal{S}) \geq \tfrac{1}{2}\rev{}(F^*)$ or $\rev{}(F^* \cap (E\setminus\mathcal{S})) \geq \tfrac{1}{2}\rev{}(F^*)$ holds. By Lemmata \ref{l.non-skeleton} and \ref{l.skeleton}, the better solution of $F^{E\setminus\mathcal{S}}$ and $F^{\mathcal{S}}$ is a $\tfrac{1}{128}$-approximation in expectation.
	\end{proof}
	
	\subsection{Skeleton Algorithm---Proof of Section~\ref{l.skeleton}}
	\label{subsec:skeleton}
	Let $F^*$ be some optimal solution. We consider an instance of FZA in which  every commodity $i\in [k]$ is separated by some given tree decomposition $\mathcal{T} = \{T_1,\ldots,T_d\}$ of the underlying tree $T=(V,E)$. Thus, the path of every commodity  $i \in [k]$ contains at least one border vertex. Given the skeleton $\mathcal{S}$ of tree decomposition $\mathcal{T}$, we call an inclusion-wise maximal subpath $\sigma$ of $\mathcal{S}$ which contains neither border nor junction vertices as inner vertices a \emph{segment} of $\mathcal{S}$.
	For example, the skeleton of the tree decomposition depicted in Figure~\ref{fig:skeleton} decomposes into six segments. Let $\Sigma(\mathcal{S})$ denote the collection of all segments of $\mathcal{S}$. It holds that $|\Sigma(\mathcal{S})| \leq |V_J| + |V_B|$ where $V_J$ and $V_B$ denote the sets of junction and border vertices, respectively. As $|V_J| < |V_B| < d$, we get $|\Sigma(\mathcal{S})| \in \mathcal{O}(\sqrt{\log n})$.
	
	\paragraph{Description of the algorithm.}
	The algorithm starts by guessing the approximate number of cuts $|F^*\cap \sigma|$  that $F^*$ makes on every segment $\sigma\in \Sigma(\mathcal{S})$. That is, the algorithm iterates over all vectors
	\[
	p\in \Gamma:=\left\{p\in \mathbb{Z}_+^{\Sigma(\mathcal{S})} \ \Big| \ p(\sigma)= 2^\kappa \mbox{ for some } \kappa\in \{0,1, \ldots  \lfloor \log_2 |\sigma| \rfloor\} \ \forall \sigma \in \Sigma(\mathcal{S})
	\right\}.
	\]
	Note that there exists a vector $p\in \Gamma$ such that  $\frac{1}{2} |\sigma\cap F^*| < p(\sigma)\le |\sigma\cap F^*|$ for every $\sigma \in \Sigma(\mathcal{S})$. 
	For each vector $p\in \Gamma$, the algorithm iterates over all segments $\sigma \in \Sigma(\mathcal{S})$ and computes a cut set $F(\sigma, p)\subseteq \sigma$ using the subroutine described below. Together, these sets form the solution  $F(p) \coloneqq \bigcup_{\sigma\in \Sigma(\mathcal{S})} F(\sigma,p)$. After evaluating $F(p)$ for every $p \in \Gamma$, the algorithm returns the cut set $F^{\mathcal{S}}$ of maximum revenue among the computed solutions, i.e., $F^{\mathcal{S}} \in \argmax\left\{\rev{}(F(p))\mid p\in \Gamma\right\}$.
	
	For every commodity $i$, we call all segments with non-empty intersection with $P_i$ \emph{segments of $i$}. A segment $\sigma$ of $i$ is an \textit{inner segment} of $i$, if $\sigma \subseteq P_i$, otherwise $\sigma$ is an \textit{outer segment} of $i$. Note that, given  some fixed $p \in \Gamma$ together with a solution $F$ which places exactly $p(\sigma)$ cuts on every segment $\sigma \in \Sigma(\mathcal{S})$, the size of $F \cap P_i$ only depends on the intersection of $F$ with the outer segments of $i$. Our algorithm utilizes this by computing good solutions with a fixed number of cuts on each segment. 
	
	\paragraph{DP for solving generalized rooted instances.} To compute the solution $F(\sigma, p)$,  we adapt the dynamic program for rooted instances (cf. Section \ref{sec:singlesource}) to a variant of FZA, which we call \emph{generalized rooted \ac{FZA} on a path}. As the name suggests, the problem restricts to paths $T=(V,E)$ with a designated endpoint, called the \emph{root} of $T$, which also is an endpoint of every commodity. In contrast to standard rooted \ac{FZA}, generalized \ac{FZA} allows for commodity-specific monotone subadditive revenue functions $f_i\colon \mathbb{N}\to \mathbb{R}_+$. 
	In Appendix~\ref{app:GeneralizedRooted}, we describe and analyze a polynomial algorithm that computes a revenue-maximizing cut set $F\subseteq E$, under the additional restriction that $|F|=p$ for some number $p\in \mathbb{Z}_+$ specified in the input.
	
	\paragraph{Subroutine to compute $F(\sigma, p)$.} 
	Given a segment $\sigma$ with endpoints $t_\sigma^{(1)}$, $t_\sigma^{(2)}$ and an integer $p(\sigma)$, the subroutine computes $F(\sigma, p)$ by randomizing between the following three options with probability $(\frac{1}{2}, \frac{1}{4},\frac{1}{4})$, respectively:
	\begin{itemize}
		\item Set  $F(\sigma, p)=\emptyset$. In this case, we say that $\sigma$ is \emph{inactive}, otherwise $\sigma$ is called \emph{active}.
		\item  Compute an optimal solution $F(\sigma, p)$ of the generalized rooted instance $\mathcal{I}(\sigma, t_{\sigma}^{(1)}, p(\sigma))$.
		\item  Compute an optimal solution $F(\sigma, p)$ of the generalized rooted instance $\mathcal{I}(\sigma, t_{\sigma}^{(2)}, p(\sigma))$.
	\end{itemize}
	
	Given some segment $\sigma$, a prescribed number of cuts $p(\sigma)$, and a root vertex $r\in \{t_{\sigma}^{(1)},  t_{\sigma}^{(2)}\}$, we define the generalized rooted instance $\mathcal{I}(\sigma, r,  p(\sigma))$ on graph $\sigma$ as follows.  A commodity $i$ of the original instance is present in  $\mathcal{I}(\sigma, r,  p(\sigma))$ if  all of the following conditions hold: 
	(a) $r$ is an inner vertex of $\pathGraph_i$,
	(b) $\pathGraph_i$ does not fully contain $\sigma$, and 
	(c) the outer segment of $i$ that does not contain $r$ (if it exists) is inactive.  
	For each commodity $i$ present in $\mathcal{I}(\sigma, r,  p(\sigma))$, we restrict its path to $\pathGraph_i^{\sigma} \coloneqq \pathGraph_i\cap \sigma$, and decrease its upper bound to $u_i^{\sigma} \coloneqq u_i - z_i^{\text{inner}}$, where
	$
	z_i^{\text{inner}} \coloneqq \sum \{p(\sigma^*)\mid \sigma^* \text{ is an active inner segment of } \pathGraph_i\}.
	$
	Furthermore, we define the commodity-specific subadditive pricing function $f_i^{\sigma}(n) \coloneqq f_i(n + z_i^{\text{inner}})$.
	
	\paragraph{Polynomial running time.}
	Note  that $|\Gamma|\in \mathcal{O}(\log n)^{\mathcal{O}(\sqrt{\log n})}=o(n)$. Thus, the algorithm considers at most $o(n)$ many vectors $p$ when iterating over $\Gamma$.  Each candidate price vector $p\in \Gamma$ requires $|\Sigma(\mathcal{S})|\in \mathcal{O}(\sqrt{\log n})$ calls to the subroutine (one for each segment $\sigma$). Since the subroutine runs in polynomial time, the algorithm for the skeleton case remains polynomial.
	
	\paragraph{Approximation guarantee.} 
	To show that $\mathbb{E}[\rev{}(F^{\mathcal{S}})] \geq \frac{1}{64}\rev{}(F^*)$, where $F^*$ is some optimal solution, we consider the iteration of the algorithm in which a cost vector $p\in \Gamma$ is evaluated for which $\frac{1}{2} |\sigma\cap F^*| < p(\sigma)\le |\sigma\cap F^*|$ for every $\sigma \in \Sigma(\mathcal{S})$.
	Let $M^*$ be the set of commodities that are served by $F^*$, i.e., the set of all $i\in [k]$ for which $|F^* \cap P_i| \leq u_i$. We partition $M^*$ into \emph{inner} and \emph{outer commodities} $M^*=I\cup O$ based on the location of the majority of cuts from $F^*$ on $P_i$. \cref{fig:inner_and_outer_commodities} illustrates the definition of $I$ and $O$ below.
	
	\begin{definition}[Inner and outer commodities]
		A commodity $i \in M^*$ is an \emph{inner commodity}, if $F^*$ places at least as many cuts on the inner segments than on the (at most two) outer segments of $P_i$. Otherwise, we call $i$ an \emph{outer commodity}. We denote the sets of inner and outer commodities by $I$ and $O$, respectively.
	\end{definition}
	
	\begin{figure}
		\centering
		\begin{tikzpicture}
\draw[line width = 1pt, gray!40, dashed] (0,.5) -- (0, -3.3);
\draw[line width = 1pt, gray!40, dashed] (2,.5) -- (2, -3.3);
\draw[line width = 1pt, gray!40, dashed] (4,.5) -- (4, -3.3);
\draw[line width = 1pt, gray!40, dashed] (6,.5) -- (6, -3.3);
\draw[line width = 1pt, gray!40, dashed] (8,.5) -- (8, -3.3);
\draw[line width = 1pt, gray!40, dashed] (10,.5) -- (10, -3.3);

\node[draw, fill = black, scale = 0.35, shape = circle] (A) at (0,-.5) {};
\node[draw, fill = black, scale = 0.35, shape = circle] (B) at (2,-.5) {};
\node[draw, fill = black, scale = 0.35, shape = circle] (C) at (4,-.5) {};
\node[draw, fill = black, scale = 0.35, shape = circle] (D) at (6,-.5) {};
\node[draw, fill = black, scale = 0.35, shape = circle] (E) at (8,-.5) {};
\node[draw, fill = black, scale = 0.35, shape = circle] (F) at (10,-.5) {};

\draw[line width = 2pt] (A) -- (B);
\draw[line width = 1.5pt, densely dotted, gray] (B) -- (C);
\draw[line width = 2pt] (C) -- (D);
\draw[line width = 2pt] (D) -- (E);
\draw[line width = 1.5pt, densely dotted, gray] (E) -- (F);

\draw[line width = 1.5pt, rwthmagenta] (1,-.2) -- (1,-3.3);
\draw[line width = 1.5pt, rwthmagenta] (3,-.2) -- (3,-3.3);
\draw[line width = 1.5pt, rwthmagenta] (3.5,-.2) -- (3.5,-3.3);
\draw[line width = 1.5pt, rwthmagenta] (5,-.2) -- (5,-3.3);
\draw[line width = 1.5pt, rwthmagenta] (6.6,-.2) -- (6.6,-3.3);
\draw[line width = 1.5pt, rwthmagenta] (7.2,-.2) -- (7.2,-3.3);
\draw[line width = 1.5pt, rwthmagenta] (7.6,-.2) -- (7.6,-3.3);
\draw[line width = 1.5pt, rwthmagenta] (8.25,-.2) -- (8.25,-3.3);
\draw[line width = 1.5pt, rwthmagenta] (9,-.2) -- (9,-3.3);
\draw[line width = 1.5pt, rwthmagenta] (9.5,-.2) -- (9.5,-3.3);

\draw[line width = 1.5pt, rwthblue] (0.5, -1) -- (6.9, -1);
\draw[line width = 1.5pt, rwthblue] (0.5, -0.85) -- (0.5, -1.15);
\draw[line width = 1.5pt, rwthblue] (6.9, -0.85) -- (6.9, -1.15);

\draw[line width = 1.5pt, rwthblue] (2.5, -1.5) -- (8.75, -1.5);
\draw[line width = 1.5pt, rwthblue] (2.5, -1.35) -- (2.5, -1.65);
\draw[line width = 1.5pt, rwthblue] (8.75, -1.35) -- (8.75, -1.65);

\draw[line width = 1.5pt, rwthblue] (4.5, -2) -- (6.9, -2);
\draw[line width = 1.5pt, rwthblue] (4.5, -1.85) -- (4.5, -2.15);
\draw[line width = 1.5pt, rwthblue] (6.9, -1.85) -- (6.9, -2.15);

\draw[line width = 1.5pt, rwthblue] (4.5, -2.5) -- (9.75, -2.5);
\draw[line width = 1.5pt, rwthblue] (4.5, -2.35) -- (4.5, -2.65);
\draw[line width = 1.5pt, rwthblue] (9.75, -2.35) -- (9.75, -2.65);

\draw[line width = 1.5pt, rwthblue] (6.3, -3) -- (8.75, -3);
\draw[line width = 1.5pt, rwthblue] (6.3, -2.85) -- (6.3, -3.15);
\draw[line width = 1.5pt, rwthblue] (8.75, -2.85) -- (8.75, -3.15);

\draw[->, line width = 1pt, rwthlightblue] (1.85,-.4) -- (1.85,0) -- (1.35, 0);

\draw[->, line width = 1pt, rwthlightblue] (4.15,-.4) -- (4.15,0) -- (4.65, 0);

\draw[->, line width = 1pt, rwthlightblue] (7.85,-.4) -- (7.85,0) -- (7.35, 0);

\node[draw = none] (L1) at (0.25, -1) {\footnotesize{\textcolor{rwthblue}{$i_1$}}};
\node[draw = none] (L2) at (2.25, -1.5) {\footnotesize{\textcolor{rwthblue}{$i_2$}}};
\node[draw = none] (L3) at (4.25, -2) {\footnotesize{\textcolor{rwthblue}{$i_3$}}};
\node[draw = none] (L4) at (4.25, -2.5) {\footnotesize{\textcolor{rwthblue}{$i_4$}}};
\node[draw = none, circle, fill = white, inner sep = 0pt, outer sep = 0pt, minimum size = 0cm] (L5) at (6.05, -3){\footnotesize{\textcolor{rwthblue}{$i_5$}}};

\node[draw = none] (S1) at (1,0.3) {\footnotesize{$p(\sigma_1) = 1$}};
\node[draw = none] (S2) at (3,0.3) {\footnotesize{$p(\sigma_2) = 2$}};
\node[draw = none] (S3) at (5,0.3) {\footnotesize{$p(\sigma_3) = 1$}};
\node[draw = none] (S4) at (7,0.3) {\footnotesize{$p(\sigma_4) = 2$}};
\node[draw = none] (S5) at (9,0.3) {\footnotesize{$p(\sigma_5) = 2$}};
\end{tikzpicture}
		
		\caption{The skeleton algorithm illustrated on an excerpt of a skeleton whose edges are represented by five segments. Segments $\sigma_2$ and $\sigma_5$ were deactivated in the randomization step, for the remaining segments, the randomly chosen root for the subroutine is indicated by an arrow. Vertical lines represent a (hypothetical) optimal solution $F^*$. Horizontal lines below represent commodities. Among those, $i_1$ and $i_2$ are inner commodities, but only $i_2$ is a nice, inner commodity. The remaining commodities are outer commodities;
			$i_3$ is not nice as both of its outer segments are active, $i_4$ is not nice since the root chosen in $\sigma_3$ for the rooted algorithm is not part of $P_4$, and $i_5$ is nice.
		} 
		
		\label{fig:inner_and_outer_commodities}
	\end{figure}
	
	For both sets $I$ and $O$, we specify a subset of \emph{nice} commodities and show in Observations~\ref{obs:inner_nice_prob} and \ref{obs:outer_nice_prob} that every commodity is nice with probability at least $\tfrac{1}{8}$. Moreover, we show in Lemmata~\ref{lem:skeleton_inner_2} and \ref{lem:skeleton_outer_2} that the algorithm obtains at least $\tfrac{1}{8}$ of the revenue which $F^*$ obtains from the set of nice commodities. 
	
	\begin{definition}[Nice, inner commodity]
		\label{def:nice_inner_commodity}
		We say that an inner commodity $i \in I$ is \emph{nice}, and write $i \in\tilde{I}$, if all of the (at most two) outer segments of $i$ are inactive, and additionally
		\begin{equation}\label{cond2}
			\sum_{\substack{\sigma \text{ active inner} \\ \text{segment of } P_i}} p(\sigma) \geq \sum_{\substack{\sigma \text{ inactive inner} \\ \text{segment of } P_i}} p(\sigma).
		\end{equation}
	\end{definition}
	
	\begin{observation}
		\label{obs:inner_nice_prob}
		Every inner commodity $i \in I$ is nice with probability at least $\tfrac{1}{8}$.
	\end{observation}
	
	\begin{proof}
		As the options chosen for every segment of $i$ in the randomization process are independent of each other (cf. subroutine to compute $F(\sigma,p)$), all outer segments of $i$ are inactive with probability at least $\tfrac{1}{4}$.  Note that condition (\ref{cond2}) holds with probability at least $1/2$. To see this, fix a subset $S$ of the inner segments of $i$. The probability that exactly the segments in $S$ are active is the same as the probability that exactly the segments in $S$ are inactive. For at least one of these options condition (\ref{cond2}) holds. 
	\end{proof}
	
	\begin{lemma}
		\label{lem:skeleton_inner_2}
		Every nice, inner commodity $i \in \tilde{I}$ satisfies $\rev{i}(F^{\mathcal{S}}) \geq \tfrac{1}{8}\rev{i}(F^* \cap \mathcal{S})$.
	\end{lemma}
	
	\begin{proof}
		Let $i$ be a nice, inner commodity, i.e., $i \in \tilde{I}$. By definition, $F^*$ serves $i$,  no outer segment of $i$ is active, and $p(\sigma) \leq |F^* \cap \sigma|$ for all $\sigma \in \Sigma(\mathcal{S})$. Hence, $F^* \cap \mathcal{S}$ makes at least as many cuts on $\pathGraph_i$ as $F^{\mathcal{S}}$. We conclude that $F^{\mathcal{S}}$ also serves $i$. 
		To gauge the number of cuts in $F^{\mathcal{S}} \cap P_i$, note that at most half of the cuts in $(F^* \cap \mathcal{S}) \cap P_i$ are on outer segments of $P_i$. As $i$ is nice, the active inner segments account for at least half of the cuts made on all inner segments, and for these segments we have $|F^{\mathcal{S}} \cap \sigma| \geq \tfrac{1}{2}|F^* \cap \sigma|$. Thus, $|F^{\mathcal{S}} \cap \pathGraph_i| \geq \tfrac{1}{8}|F^* \cap \pathGraph_i|$. As the pricing function $f$ is subadditive, the statement follows.
	\end{proof}
	
	Similarly to \cref{def:nice_inner_commodity}, we define a set $\tilde{O}$ of nice, outer commodities. To do so, we distinguish whether $i$ has one or two outer segments in $\mathcal{S}$. \cref{fig:inner_and_outer_commodities} visualizes these definitions.
	
	\begin{definition}[Nice, outer commodity]
		\label{def:nice_outer_commodity}
		Let $i \in O$ be an outer commodity. We call $i$ \emph{nice} (and write $i \in \tilde{O}$) if either of the following two conditions holds.
		\begin{enumerate}[a)]
			\item Commodity $i$ only has one outer segment $\sigma$, that segment is active, and the root chosen to compute $F(\sigma,p)$ is a vertex in $\pathGraph_i$.
			\item Commodity $i$ has two outer segments $\sigma_1$ and $\sigma_2$, and the following three conditions hold.
			\begin{enumerate}[(i)]
				\item Exactly one of $\sigma_1$ and $\sigma_2$ is active, w.l.o.g. $\sigma_1$ is active.
				\item $|F^* \cap \sigma_1| \geq |F^* \cap \sigma_2|$.
				\item The root chosen to compute $F(\sigma_1,p)$ is a vertex in $\pathGraph_i$.
			\end{enumerate}
		\end{enumerate}
	\end{definition}
	
	\begin{observation}
		\label{obs:outer_nice_prob}
		Every outer commodity $i \in O$ is nice with probability at least $\tfrac{1}{8}$.
	\end{observation}
	\begin{proof}
		Let $i \in O$ be arbitrary. As each (of the one or two) outer segments of $i$ is active with probability $1/2$, with probability $1/2$ exactly one of the outer segments is active. If there are two outer segments, with probability $1/2$, the active outer segment $\sigma$ is the one on which $F^*$ places more cuts. Lastly, with probability $1/2$ the subroutine to compute $F(\sigma,p)$ chooses the endpoint of $\sigma$ that lies in $P_i$ as root.
	\end{proof}
	
	To prove the approximation guarantee for outer segments, we introduce another Lemma, which we prove in Appendix~\ref{app:skeleton}.
	
	\begin{restatable}{lemma}{sublogSubadditivity}
		\label{lem:sublog_submodular_2}
		Let $\rev{}(F)$ be the revenue attained by a solution $F$ with $|F| = m \geq 1$. There exists a solution $\bar{F} \subseteq F$ with $|\bar{F}| = 2^{\lfloor \log_2 m \rfloor}$ with $\rev{}(\bar{F}) \geq \rev{}(F)/2$.
	\end{restatable}
	
	By definition, every nice, outer commodity has exactly one active outer segment. Based on this, we partition the set of nice outer commodities $\tilde{O}$ into sets
	\[\tilde{O}_{\sigma} \coloneqq \{i \in \tilde{O} \ | \ \sigma \text{ is the active outer segment of }\pathGraph_i\}, \quad \sigma \in \Sigma(\mathcal{S}).\]
	We conclude the case of outer commodities by showing that $F^\mathcal{S}$ obtains a constant fraction of the revenue that $F^* \cap \mathcal{S}$ obtains from all commodities in $\tilde{O}_{\sigma}$ for every $\sigma \in \Sigma(\mathcal{S})$. Recall that for every active segment the algorithm computes a cut set $F(\sigma, p)\subseteq \sigma$ using the DP described in Appendix~\ref{app:GeneralizedRooted}, on the generalized rooted instance $\mathcal{I}(\sigma, p, r)$ for $r\in \{t_{\sigma}^{(1)}, t_{\sigma}^{(2)}\}.$ We denote the set of commodities that are present in this instance by $M_{\sigma}$.
	
	\begin{lemma}
		\label{lem:skeleton_outer_2}
		For every $\sigma \in \Sigma(\mathcal{S})$ it holds that $\rev{M_{\sigma}}(F^{\mathcal{S}}) \geq \tfrac{1}{8}\rev{\tilde{O}_{\sigma}}(F^* \cap \mathcal{S})$.
	\end{lemma}
	
	\begin{proof}
		Fix some $\sigma \in \Sigma(\mathcal{S})$. For a given subset  $M' \subseteq [k]$ of commodities, and some given number $p' \in \mathbb{N}$, let $\DP^{\sigma, p'}_{M'}$ be the optimal objective value of the DP for generalized rooted instances, when called on commodity set $M'$, segment $\sigma$, and  predefined number $p'$ of cuts.
		
		First, note that $\tilde{O}_{\sigma} \subseteq M_{\sigma}$, and that the sets $\tilde{I}$ and $M_{\sigma}$ are pairwise disjoint, since for nice inner commodities all outer segments (one of which would be $\sigma$) are inactive. Recall that in the generalized rooted instance  $\mathcal{I}(\sigma, p, r)$, we assign to each commodity $i \in M_{\sigma}$ an individual pricing  function that accounts for the number of cuts made on the inner segments of $i$. Thus, since for all commodities $i \in M_{\sigma}$ the other outer segment is inactive, we have $\rev{M_{\sigma}}(F^\mathcal{S}) = \DP_{M_{\sigma}}^{\sigma, p(\sigma)}$. That is, the objective value of the $DP$ when called with set of commodities $M_{\sigma}$ and individual pricing  functions $f_i^{\sigma}$ using only the edges in segment $\sigma$, is the same as the revenue obtained by the algorithm from commodities in $M_{\sigma}$ with the original pricing  function $f$ when accounting for all cuts made, i.e., even those outside of $\sigma$. Hence we obtain the following chain of inequalities
		\[\rev{M_{\sigma}}(F^{\mathcal{S}}) = \DP_{M_{\sigma}}^{\sigma, p(\sigma)} \overset{(1)}{\geq} \DP_{\tilde{O}_{\sigma}}^{\sigma, p(\sigma)} \overset{(2)}{\geq} \frac{1}{2}\DP_{\tilde{O}_{\sigma}}^{\sigma, |F^* \cap \sigma|} \overset{(3)}{\geq} \frac{1}{8}\rev{\tilde{O}_{\sigma}}(F^*).\]
		Inequality $(1)$ uses that $\tilde{O}_{\sigma} \subseteq M_{\sigma}$. Inequality $(2)$ holds by \cref{lem:sublog_submodular_2}, and inequality $(3)$ uses the fact that $f$ is subadditive and for all commodities in $\tilde{O}_{\sigma}$, the optimal solution $F^*$ places at least $1/4$ of its cuts on $\pathGraph_i$ on $\sigma$. This concludes the proof.
	\end{proof}
	
	Plugging in the bounds obtained for inner and outer commodities, we now prove \cref{l.skeleton}.
	
	\begin{proof}[Proof (of \cref{l.skeleton})]
		In \cref{lem:skeleton_inner_2}, we lower bound the revenue obtained from nice, inner commodities, and in \cref{lem:skeleton_outer_2} we lower bound the profit obtained by commodities in $M_{\sigma}$ for $\sigma \in \Sigma(\mathcal{S})$. Note that these sets of commodities are disjoint, as no commodity in $\tilde{I}$ has any active outer segments, but every commodity in some $M_{\sigma}$ does. Thus, we obtain
		\[\rev{}(F^{\mathcal{S}}) \geq \frac{1}{8}\rev{\tilde{I}}(F^* \cap \mathcal{S}) + \sum_{\sigma\in\Sigma(\mathcal{S})} \frac{1}{8}\rev{\tilde{O}_{\sigma}}(F^* \cap \mathcal{S})
		= \frac{1}{8}\rev{\tilde{I}}(F^* \cap \mathcal{S}) + \frac{1}{8}\rev{\tilde{O}}(F^* \cap \mathcal{S}).\]
		As every (inner or outer) commodity served by $F^*$ is nice with probability at least $1/8$, linearity of expectation implies $\mathbb{E}[\rev{}(F^{\mathcal{S}})] \geq \frac{1}{64}\rev{}(F^* \cap \mathcal{S})$.
	\end{proof}
	
	\bibliography{FareZones.bib}
	
	
	\clearpage
	
	\appendix
	\FloatBarrier
	\section{Hardness Results}
	\label{app:Hardness}
	In this section, we discuss the computational complexity of \ac{FZA}. In \cref{subsec:APX_stars}, we show \textsf{APX}-hardness of \ac{FZA} when restricted to star graphs, and then, in \cref{subsec:NP_paths}, we show strong \textsf{NP}-hardness of \ac{FZA} when restricted to paths. For both results, we use a reduction from \textsc{Max 2-Sat (3)} (see \cref{def:Max2Sat3}) which is \textsf{APX}-hard due to Ausiello et al.~\cite{Max2SatAPX2003}. Note that in \cref{app:ptas} we sketch a PTAS for \ac{FZA} when restricted to paths. Thus, there is no hope to extend the \textsf{APX}-hardness result presented in \cref{thm:APX-hardness_Stars} to paths. 
	
	\begin{definition}[\textsc{Max 2-Sat (3)}]
		\label{def:Max2Sat3}
		Input: A Boolean 2-CNF formula $\varphi$ with variables $x_1,\ldots,x_n$ and clauses $\clause_1,\ldots,\clause_m$ such that each variable appears in at most three clauses.
		\newline
		Objective: Find the maximum integer $y$ and a variable assignment $\mathfrak{J}: \{x_1, \dots, x_n\} \to \{\texttt{TRUE},\texttt{FALSE}\}$ such that $\mathfrak{J}$ satisfies $y$ clauses of $\varphi$.
	\end{definition}
	
	We remark that Otto and Boysen \cite{Otto2017} show \textsf{NP}-hardness on star graphs for a closely related problem, however, their reduction heavily relies on the possibility to restrict the maximum number of zones by an arbitrary positive integer. As this is not possible in \ac{FZA}, it is unclear how to adapt the construction from \cite{Otto2017}.
	
	\subsection{\textsf{APX}-hardness on Star Graphs}
	\label{subsec:APX_stars}
	
	We first prove \textsf{APX}-hardness on star graphs. We first describe our construction and then show that it is approximation-preserving.
	
	\begin{restatable}{theorem}{APXhardnessStars}
		\label{thm:APX-hardness_Stars}
		\ac{FZA} is \textsf{APX}-hard when restricted to instances on star graphs with maximum congestion 3.
	\end{restatable}
	
	\begin{proof}
		Given a Boolean formula $\varphi$ with variables $x_1,...,x_n$, clauses $\clause_1,...,\clause_m$ (consisting of two literals each and with each variable appearing in at most 3 clauses), and a clause-target $y \in \mathbb{Z}_{\geq 0}$, we construct an instance of \ac{FZA} on a star. Figure~\ref{fig:APX_Stars} depicts an example of such an instance which is constructed as follows:
		\begin{itemize}
			\item Fix the pricing function $f(x) = x$.
			\footnote{The specific choice of $f$ is not important for the reduction. One can accommodate many other functions by adjusting the weights accordingly.}
			\item Construct a vertex $v$ serving as the center of the star.
			\item For each variable $x_i$, add two vertices $x_i$ and $\bar{x}_i$ and connect them to $v$.
			\item For each variable $x_i$, add a commodity $i$ with a path from $x_i$ to $\bar{x}_i$, budget $\ub_i = 1$ and $\weight_i = 6$. We call these commodities \textit{variable commodities}.
			\item For each clause $\clause_j = \{l_j^{(1)}, l_j^{(2)}\}$, add three parallel \textit{clause commodities} $\ell_j^{(1)}$, $\ell_j^{(2)}$ and $\ell_j^{(3)}$, all with $\ub_{\ell_j^{(i)}} = w_{\ell_j^{(i)}} = 1$. The path of commodity $\ell_j^{(1)}$ connects $v$ to $l_j^{(1)}$, and the path of commodity $\ell_j^{(2)}$ connects $v$ to $l_j^{(2)}$. Finally, commodity $\ell_j^{(3)}$ connects $l_j^{(1)}$ to $l_j^{(2)}$.
			\item We define the target revenue for our \ac{FZA}-instance as $z = 6n + 2y$.
		\end{itemize}
		
		\begin{figure}[tbp]
			\centering
			\resizebox{!}{.3\textwidth}{
				\resizebox{.85\textwidth}{!}{
	\begin{tikzpicture}[
		every node/.style={draw=black, circle, fill=none, minimum size=.4cm, inner sep=0cm, outer sep=0cm, font = \footnotesize, opacity = 0.9}, scale = 1
		]
		
		\pgfmathsetmacro{\singleoffset}{0.3}
		\pgfmathsetmacro{\doubleoffset}{0.45}
		\pgfmathsetmacro{\tripleoffset}{.6}
		\pgfmathsetmacro{\quadrupleoffset}{.75}
		\pgfmathsetmacro{\quintupleoffset}{.9}
		\pgfmathsetmacro{\sextupleoffset}{1.05}
		
		\begin{pgfonlayer}{fg}    
			\node[fill = white, inner sep= .7 5pt] (center) at (0,0) {\small $v$};
			\node[fill = white, inner sep= .7 5pt, above = 1.4cm of center] (x1) {\scriptsize $\varvertex_1$};
			\node[fill = white, inner sep= .7 5pt, right = 1.4cm of center] (barx1) {\scriptsize $\bar{\varvertex}_1$};
			\node[fill = white, inner sep= .7 5pt, below = 1.4cm of center] (x2) {\scriptsize $\varvertex_2$};
			\node[fill = white, inner sep= .7 5pt, left = 1.4cm of center] (barx2) {\scriptsize $\bar{\varvertex}_2$};
			\path[draw, thick] (x1) -- (center) -- (barx1);
			\path[draw, thick] (x2) -- (center) -- (barx2);
			\path[very thick, draw = rwthblue, dashed] 
			($ (x1) + (\singleoffset, 0) $) 
			-- ($ (center) + (\singleoffset, \singleoffset) $)
			-- ($ (barx1) + (0, \singleoffset) $); 
			\path[very thick, draw = rwthblue, dashed] 
			($ (x2) + (-\singleoffset, 0) $) 
			-- ($ (center) + (-\singleoffset, -\singleoffset) $)
			-- ($ (barx2) + (0, -\singleoffset) $);
		\end{pgfonlayer}
		\draw[thick, color = rwthlightblue, line width=1.8pt]
		($ (barx2) + (0, \doubleoffset) $)
		-- ($ (center) + (-\doubleoffset, \doubleoffset) $)
		-- ($ (x1) + (-\doubleoffset, 0) $);
		\draw[thick, color = rwthlightblue, line width=1.8pt]
		-- ($ (barx2) + (0, \singleoffset) $)
		-- ($ (center) + (-\singleoffset, \singleoffset) $);
		\draw[thick, color = rwthlightblue, line width=1.8pt]($ (center) + (-\tripleoffset, \tripleoffset) $)
		-- ($ (x1) + (-\tripleoffset, 0) $);
		\draw[thick, color = rwthmagenta, line width=1.8pt]
		($ (barx2) + (0, -\doubleoffset) $)
		-- ($ (center) + (0, -\doubleoffset) $)
		-- ($ (barx1) + (0, -\doubleoffset) $);
		\draw[thick, color = rwthmagenta, line width=1.8pt]
		-- ($ (center) + (0, -\singleoffset) $)
		-- ($ (barx1) + (0, -\singleoffset) $);
		\draw[thick, color = rwthmagenta, line width=1.8pt]
		($ (barx2) + (0, -\tripleoffset) $)
		-- ($ (center) + (0, -\tripleoffset) $);
	\end{tikzpicture}
}
			}
			\caption{The construction used in the proof of Theorem~\ref{thm:APX-hardness_Stars} for $\phi = \textcolor{rwthlightblue}{(x_1 \lor \bar{x}_2)} \land \textcolor{rwthmagenta}{(\bar{x}_1 \lor \bar{x}_2)}$. Dashed lines indicate \textit{variable commodities}. The remaining commodities represent \textit{clause commodities}.}
			\label{fig:APX_Stars}
		\end{figure}
		
		We now show that the constructed \ac{FZA}-instance admits a solution $\cutset$ that cuts precisely one of the edges $\{x_i,v\}$ and $\{\bar{x}_i,v\}$ for each $i \in [n]$, and the revenue generated by $\cutset$ reaches the target $z$ if and only if the variable assignment 
		\[
		\mathfrak{J}\colon x_i \mapsto 
		\begin{cases}
			\texttt{TRUE}, &\text{ if } \{x_i, v\} \in \cutset,\\
			\texttt{FALSE}, &\text{ if } \{\bar{x}_i, v\} \in \cutset
		\end{cases}
		\]    
		satisfies at least $y$ clauses of $\phi$.
		First, we observe the following.
		\begin{observation}\label{obs:StarReduction_ClauseRevenue}
			For each clause $\clause_j = \{l_j^{(1)}, l_j^{(2)}\}$ of $\phi$, the three clause commodities $\ell_j^{(1)}$, $\ell_j^{(2)}$ and $\ell_j^{(3)}$  corresponding to $\clause_j$ yield a total revenue of 2 if at least one of the edges $\{l_j^{(1)}, v\}$ and $\{l_j^{(2)}, v\}$ is cut by some solution $F$, otherwise $F$ obtains no revenue from $\ell_j^{(1)}$, $\ell_j^{(2)}$ and $\ell_j^{(3)}$.
		\end{observation}
		\begin{proposition}
			\label{prop:StarReduction_ValidAssignment}
			Given a solution $F$ to the constructed \ac{FZA}-instance, one can (in polynomial time) construct a solution $F'$ with $\rev{}(F') \geq \rev{}(F)$ such that $F'$ places exactly one cut on every variable commodity.
		\end{proposition}
		\begin{proof}
			Assume for contradiction that there exists a variable commodity $i$ corresponding to variable $x_i$ such that 
			\begin{enumerate}
				\item $\cutset \cap P_i = \emptyset$, or \label{enum:HardnessStars1}
				\item $\cutset \cap P_i = P_i$. \label{enum:HardnessStars2}
			\end{enumerate}
			If there is no such $i$, the original solution $F$ already has the desired property. 
			
			If there exists an $i$ that satisfies the first condition, consider the set $F' = \cutset \cup \{v, x_i\}$.
			It holds that $\rev{i}(F') = \rev{i}(\cutset) + 6$. Now, consider some clause $\clause_j$ featuring $x_i$ and the three commodities $\ell_j^{(1)}, \ell_j^{(2)}$, and $\ell_j^{(3)}$ representing $\clause_j$. By Observation~\ref{obs:StarReduction_ClauseRevenue},   
			\[
			\sum_{i=1}^3 \rev{\ell_j^{(i)}}(F') \geq \sum_{i=1}^3 \rev{\ell_j^{(i)}}(\cutset).
			\]
			Thus $\rev{}(F') \geq \rev{}(F)$ and we iterate the procedure with $F'$.
			
			Now, assume that the second condition holds for some variable $i$. Consider the solution $\cutset' = \cutset \setminus \{x_i, v\}$. It holds that $\rev{i}(\cutset') = \rev{i}(\cutset) + 6$. 
			Furthermore, for any clause $\clause_j$ featuring $x_i$,
			\[
			\sum_{i=1}^3 \rev{\ell_j^{(i)}}(\cutset') \geq \left(\sum_{i=1}^3 \rev{\ell_j^{(i)}}(\cutset)\right)-2,
			\]
			by Observation~\ref{obs:StarReduction_ClauseRevenue}. Since $x_i$ is contained in at most three clauses, $\cutset'$ fulfills $\rev{}(F') \geq \rev{}(F)$ and we iterate the procedure with $F'$.
		\end{proof}
		Given a solution $\cutset$ to the constructed \ac{FZA}-instance that cuts exactly one of the edges $\{x_i,v\}$ and $\{\bar{x}_i,v\}$ for every $i \in [n]$, it is easy to see that the variable assignment $\mathfrak{J}$ (defined above) is feasible. To see that $\mathfrak{J}$ satisfies at least $y$ clauses whenever $\cutset$ generates a revenue of at least $6n + 2y$, note that the variable commodities always contribute $6n$ to $\rev{}(F)$. Further, by Observation~\ref{obs:StarReduction_ClauseRevenue}, every triple $\ell_j^{(1)}, \ell_j^{(2)}, \ell_j^{(3)}$ of clause commodities generates a revenue of 2 if and only if $F$ contains at least one of the edges $\{v,l_j^{(1)}\}$ and $\{v,l_j^{(2)}\}$. Otherwise, the clause commodities corresponding to clause $\clause_j$ yield no revenue at all. As a result, there is a direct correspondence between every $\clause_j$ that is satisfied by $\mathfrak{J}$ and every triple of clause commodities contributing a revenue of 2 under $F$. Any other clause commodity triple (corresponding to an unsatisfied clause) contributes no revenue under $F$.
		
		We now show that the reduction is approximation-preserving. To this end, we begin by showing the following proposition.
		
		\begin{proposition}
			\label{prop:Max2SatParameters}
			Every \textsc{Max 2-Sat (3)} instance with $n$ variables, $m$ clauses, and non-trivial target $y$ satisfies $m \geq \frac{n}{2}$ and $y \geq \frac{1}{2}m$.
		\end{proposition}
		\begin{proof}
			The first bound builds on the facts that every variable must appear in at least one clause and that every clause consists of 2 literals. Hence, there must be at least $\lceil \frac{n}{2} \rceil$ clauses.
			The second bound can be seen by observing that the variable assignment 
			\[\mathfrak{J}'\colon x_i \mapsto 
			\begin{cases}
				\texttt{TRUE}, &\text{ if } x_i \text{ appears in more clauses than } \bar{x}_i,\\
				\texttt{FALSE}, &\text{ otherwise. }
			\end{cases}
			\]
			always satisfies $\frac{m}{2}$ clauses. This is because $\mathfrak{J}'$ satisfies at least half of the clauses where $x_i$ or $\bar{x}_i$ appear, for every $i \in [n]$.
		\end{proof}
		
		Let $\varepsilon'$ be some constant such that there does not exist a $(1-\varepsilon')$-approximation for \textsc{Max 2-Sat (3)}. Such a constant exists by the \textsf{APX}-hardness of \textsc{Max 2-Sat (3)}. Suppose for contradiction that there exists a $(1-\varepsilon)$-approximation for \ac{FZA} for $\varepsilon \coloneqq \tfrac{\varepsilon'}{13}$.
		Let $\varphi$ be some instance of \textsc{Max 2-Sat (3)}, and let $\mathcal{I}$ be the instance of \ac{FZA} that is constructed from $\varphi$ using the reduction described above. We denote as $F^*$ an optimal solution of $\mathcal{I}$ and by $F$ the solution computed by the assumed $(1-\varepsilon)$-approximation algorithm for \ac{FZA}. Further, we write $\OPT$ and $\ALG$ for the revenue generated by $F^*$ and $F$, respectively.
		
		By \cref{prop:StarReduction_ValidAssignment}, we can assume w.l.o.g that both $F^*$ and $F$ place exactly one cut on each variable commodity. Hence, we can write
		\begin{align}
			\label{eq:HardnessStarsRevenueOPT}
			\OPT &= 6n + 2y',\\
			\ALG &= 6n + 2y.
		\end{align}
		where $y'$ and $y$ are the number of clauses satisfied by the variable assignments corresponding to $\cutset^*$ and $\cutset$, respectively.
		Applying \cref{prop:Max2SatParameters} to Equation~(\ref{eq:HardnessStarsRevenueOPT}), we get 
		\begin{equation}
			\label{eq:HardnessStarsObjectiveBound}
			\OPT = 6n + 2y' \leq 6(2m) + 2y' = 12m + 2y' \leq 12(2y') + 2y' = 26y'.
		\end{equation}
		
		As we assumed $F$ to be a $(1-\varepsilon)$-approximate solution, it holds that
		\[
		\frac{\ALG}{\OPT} = \frac{6n+2y}{6n+2y'} = 1 - \frac{2y'-2y}{6n + 2y'} \geq 1-\varepsilon.
		\]
		We rearrange terms and apply \cref{eq:HardnessStarsObjectiveBound} to obtain
		\[
		2(y' - y) \leq \varepsilon(6n+2y') \overset{(\ref{eq:HardnessStarsObjectiveBound})}{\leq} 26y'\cdot \varepsilon = 2y' \cdot \varepsilon'.
		\]
		Thus, $y' - y \leq y'\cdot \varepsilon'$, or equivalently $\tfrac{y}{y'} \geq (1-\varepsilon')$. Thus, the variable assignment for the formula $\varphi$ induced by $F$ is a $(1-\varepsilon')$ approximation for \textsc{Max 2-Sat (3)} on this instance. As $\varphi$ was chosen arbitrarily, this contradicts the assumption that there does not exist a $(1-\varepsilon')$-approximation for \textsc{Max 2-Sat (3)}, which concludes the proof of \cref{thm:APX-hardness_Stars}.
	\end{proof}
	
	\subsection{Strong \textsf{NP}-hardness on Paths}
	\label{subsec:NP_paths}
	We conclude this section on the complexity of \ac{FZA} by proving the following theorem, again using a reduction from \textsc{Max 2-Sat (3)}.
	\begin{restatable}{theorem}{NPhardnessPaths}
		\label{thm:NP-hardness_Paths}
		\ac{FZA} is strongly \textsf{NP}-hard even when restricted to paths.
	\end{restatable}
	
	\begin{proof}
		Let $\varphi$ be a Boolean formula with variables $x_1,...,x_n$, clauses $\clause_1,...,\clause_m$ (consisting of two literals each and with each variable appearing in at most 3 clauses), and a clause-target $y \in \mathbb{Z}_{\geq 0}$. Further, let $M$ be a large constant strictly greater than $m$. We aim to construct an instance of \ac{FZA} that admits a revenue of $29 M n + y$ if and only if there exists a variable assignment that satisfies at least $y$ clauses of $\phi$. 
		
		We fix the pricing function to $f(x) = x$. We first describe a gadget modeling a variable assignment. Then, we present a gadget that ensures the validity of a variable assignment induced by an optimal solution in the constructed \ac{FZA}-instance. 
		Finally, we present how clauses from a \textsc{Max 2-Sat} can be mapped to commodities in our construction. These three steps are illustrated in Figures~\ref{fig:PathReduction-VarAss}, \ref{fig:PathReduction-FeasAss}, and \ref{fig:PathReduction-FullExample}, respectively.
		
		\paragraph{Variable-assignment gadget.}
		Consider a literal $l = x_i$ or $l = \bar{x_i}$. The variable-assignment gadget for $l$ consists of a path on 5 edges. We subdivide these edges into three blocks \blockA, \blockB, \blockC \ consisting of one edge, three edges, and one edge, respectively. We introduce the following commodities:
		\begin{itemize}
			\item Commodity $a^{(1)}$ spanning block \blockA\ and \blockB\ with $\ub_{a^{(1)}} = 2$ and $\weight_{a^{(1)}} = M$.
			\item Commodity $a^{(2)}$ spanning block \blockB\ and \blockC\ with $\ub_{a^{(2)}} = 2$ and $\weight_{a^{(2)}} = M$.
			\item Commodity $a^{(3)}$ spanning only block \blockB\ with  $\ub_{a^{(3)}} = 3$ and $\weight_{a^{(3)}} = 4M$.
			\item Commodity $a^{(4)}$ spanning only block \blockB\ with  $\ub_{a^{(4)}} = 1$ and $\weight_{a^{(4)}} = 4M$.
			\item Commodity $a^{(5)}$ spanning only block \blockA\ with  $\ub_{a^{(5)}} = 1$ and $\weight_{a^{(5)}} = M$.
			\item Commodity $a^{(6)}$ spanning only block \blockC\ with  $\ub_{a^{(6)}} = 1$ and $\weight_{a^{(6)}} = M$.
		\end{itemize}
		\begin{figure}[tbph]
			\centering        \begin{tikzpicture}[scale = 0.8,
	vertex/.style={draw, shape = circle, inner sep = 0pt, minimum size = .15cm, fill = black}
	]
	\draw[lightgray, thick, dashed] (2, -0.5) -- (2, 5);
	\draw[lightgray, thick, dashed] (8, -0.5) -- (8, 5);

    \node[draw = none, shape = circle, outer sep = .15cm, inner sep = 0pt, minimum size = .35cm] (ldel) at (-1,0) {\small$\dots$};
	\node[vertex] (1) at (0,0) {};
	\node[vertex] (2) at (2,0) {};
	\node[vertex] (3) at (4,0) {};
	\node[vertex] (4) at (6,0) {};
	\node[vertex] (5) at (8,0) {};
	\node[vertex] (6) at (10,0) {};
    \node[draw = none, shape = circle, outer sep = .15cm, inner sep = 0pt, minimum size = .35cm] (rdel) at (11,0) {\small$\dots$};
	
	\draw (ldel) -- (1) -- (2) -- (3) -- (4) -- (5) -- (6) -- (rdel);

	\draw[thick] (6, .75) -- node[pos = 0.3, above,draw=none]{$(\textcolor{rwthblue}{1}, \textcolor{rwthmagenta}{M})$} (10,.75);
	\draw[thick] (0, .75) -- node[pos = 0.7, above,draw=none]{$(\textcolor{rwthblue}{1}, \textcolor{rwthmagenta}{M})$} (4,.75);
	\draw[thick] (2, 1.5) -- node[pos = 0.85, above,draw=none]{$(\textcolor{rwthblue}{1}, \textcolor{rwthmagenta}{4M})$} (8,1.5);
	\draw[thick] (2, 2.25) -- node[pos = 0.15, above,draw=none]{$(\textcolor{rwthblue}{3}, \textcolor{rwthmagenta}{4M})$} (8,2.25);
	\draw[thick] (0, 3) -- node[pos = 0.125, above,draw=none]{$(\textcolor{rwthblue}{2}, \textcolor{rwthmagenta}{M})$} (8,3);
	\draw[thick] (2, 3.75) -- node[pos = 0.875,above,draw=none]{$(\textcolor{rwthblue}{2}, \textcolor{rwthmagenta}{M})$} (10,3.75);
	
	\node[draw=none] (A) at (0.75, 4.9) {\blockA};
	\node[draw=none] (B) at (5, 4.9) {\blockB};
	\node[draw=none] (C) at (9.25, 4.9) {\blockC};
\end{tikzpicture}
			\caption{An illustration of the variable-assignment gadget corresponding to literal $l$ used in the proof of Theorem~\ref{thm:NP-hardness_Paths}. The parameters of each depicted commodity $a^{(i)}$ are given as $(\ub_{a^{(i)}}, \weight_{a^{(i)}})$.}
			\label{fig:PathReduction-VarAss}
		\end{figure}
		\begin{observation}
			There are only two ways to obtain the maximum profit of $14M$ from the variable-assignment gadget corresponding to literal $l$, namely:
			\begin{enumerate}
				\item Three cuts in block \blockB\ and no cuts in blocks \blockA\ and \blockC\
				\item One cut in each block \blockA, \blockB, and \blockC.
			\end{enumerate}
		\end{observation}
		
		For now, we assume that $M$ is sufficiently large such that each subpath corresponding to a variable-assignment gadget contains exactly $3$ cuts. We give the details on how to choose $M$ later. 
		In the following, we interpret literal $l$ as $\texttt{TRUE}$ if we choose the cut pattern $(0 - 3 - 0)$ w.\,r.\,t. the blocks \blockA, \blockB, and \blockC\, and otherwise, \ie with pattern $(1 - 1 - 1)$, we interpret $l$ as \texttt{FALSE}.
		
		\begin{figure}[tbph]
			\centering        \resizebox{\textwidth}{!}{
    \begin{tikzpicture}[scale = 0.7,
    		commodityLabel/.style = {
    			draw=none, fill = white, inner sep = 0pt, minimum height = .4cm, yshift = .7mm
    		}
    	]
    	\draw[thick] (0,0) -- (0,3);
    	\draw[thick] (3,0) -- (3,3);
    	\draw[thick] (6,0) -- (6,3);
    	\draw[thick] (9,0) -- (9,3);
    	\draw[thick] (12,0) -- (12,3);
    	\draw[thick] (15,0) -- (15,3);
    	\draw[thick] (18,0) -- (18,3);
    	\draw[thick, dashed, lightgray] (1,0) -- (1,3);
    	\draw[thick, dashed, lightgray] (2,0) -- (2,3);
    	\draw[thick, dashed, lightgray] (4,0) -- (4,3);
    	\draw[thick, dashed, lightgray] (5,0) -- (5,3);
    	\draw[thick, dashed, lightgray] (7,0) -- (7,3);
    	\draw[thick, dashed, lightgray] (8,0) -- (8,3);
    	\draw[thick, dashed, lightgray] (10,0) -- (10,3);
    	\draw[thick, dashed, lightgray] (11,0) -- (11,3);
    	\draw[thick, dashed, lightgray] (13,0) -- (13,3);
    	\draw[thick, dashed, lightgray] (14,0) -- (14,3);
    	\draw[thick, dashed, lightgray] (16,0) -- (16,3);
    	\draw[thick, dashed, lightgray] (17,0) -- (17,3);
        \draw[thick, dashed, lightgray] (19,0) -- (19,3);
    	\draw[very thick] (2,.5) -- node[above, commodityLabel]{$(\textcolor{rwthblue}{1}, \textcolor{rwthmagenta}{M})$} (4,.5);
    	\draw[very thick] (8,.5) -- node[above, commodityLabel]{$(\textcolor{rwthblue}{1}, \textcolor{rwthmagenta}{M})$} (10,.5);
    	\draw[very thick] (14,.5) -- node[above, commodityLabel]{$(\textcolor{rwthblue}{1}, \textcolor{rwthmagenta}{M})$} (16,.5);
    	
    	\node[draw=none, inner sep = 0cm] at (1.5, -0.5) {$x_1$};
    	\node[draw=none, inner sep = 0cm] at (4.5, -0.5) {$\bar{x_1}$};
    	\node[draw=none, inner sep = 0cm] at (7.5, -0.5) {$x_2$};
    	\node[draw=none, inner sep = 0cm] at (10.5, -0.5) {$\bar{x_2}$};
    	\node[draw=none, inner sep = 0cm] at (13.5, -0.5) {$x_3$};
    	\node[draw=none, inner sep = 0cm] at (16.5, -0.5) {$\bar{x_3}$};
        \node[draw=none, inner sep = 0cm] at (18.5, -0.5) {$\dots$};
    	\node[draw=none] at (5.5, 2.5) {\small \tikzdots};
    	\node[draw=none] at (6.5, 2.5) {\small \blockA};
    	\node[draw=none] at (7.5, 2.5) {\small \blockB};
    	\node[draw=none] at (8.5, 2.5) {\small \blockC};
    	\node[draw=none] at (9.5, 2.5) {\small \blockA};
    	\node[draw=none] at (10.5, 2.5) {\small \blockB};
    	\node[draw=none] at (11.5, 2.5) {\small \blockC};
    	\node[draw=none] at (12.5, 2.5) {\small \tikzdots};
    \end{tikzpicture}
}
			\caption{Depicted are the commodities enforcing a valid assignment in the proof of Theorem~\ref{thm:NP-hardness_Paths}. The parameters of each depicted commodity $i$ are given as $(\ub_{i}, \weight_i)$.}
			\label{fig:PathReduction-FeasAss}
		\end{figure}
		\paragraph{Feasible assignments.}
		Next, we want to ensure that for any variable $x$ appearing in $\phi$, we either set $x$ or its negation $\bar{x}$ to \texttt{TRUE}. Figure~\ref{fig:PathReduction-FeasAss} illustrates the construction described in the following. 
		First, we create one variable-assignment gadget corresponding to $x_i$ and $\bar{x}_i$ for each variable $x_i$, $1 \leq i \leq n$. We introduce these gadgets from left to right beginning with the gadgets representing $x_1$ followed by $\bar{x}_1$, $x_2$, $\bar{x}_2$, \dots, $\bar{x}_n$. The graph structure underlying these gadgets constitutes the entire path of our constructed \ac{FZA}-instance.  
		
		To ensure that $x_i$ is \texttt{TRUE} if and only if $\bar{x}_i$ is \texttt{FALSE}, we introduce a commodity $i$ for each variable $x_i$ which spans block \blockC\ of $x_i$ and block \blockA\ of $\bar{x}_i$ with $\ub_i = 1$ and $\weight_i = M$. These commodities are depicted in Figure~\ref{fig:PathReduction-FeasAss}.
		
		\begin{figure}[tbhp]
			\centering 
			\resizebox{\textwidth}{!}{
				\begin{tikzpicture}[scale = 0.7,
		commodityLabel/.style = {
			draw=none, fill = white, inner sep = 0pt, minimum height = .4cm, yshift = .7mm
		},
		blockMarker/.style = {
			thick, dashed, lightgray
		}
	]
	\draw[blockMarker] (1,0) -- (1,7.5);
	\draw[blockMarker] (2,0) -- (2,7.5);
	\draw[blockMarker] (4,0) -- (4,7.5);
	\draw[blockMarker] (5,0) -- (5,7.5);
	\draw[blockMarker] (7,0) -- (7,7.5);
	\draw[blockMarker] (8,0) -- (8,7.5);
	\draw[blockMarker] (10,0) -- (10,7.5);
	\draw[blockMarker] (11,0) -- (11,7.5);
	\draw[blockMarker] (13,0) -- (13,7.5);
	\draw[blockMarker] (14,0) -- (14,7.5);
	\draw[blockMarker] (16,0) -- (16,7.5);
	\draw[blockMarker] (17,0) -- (17,7.5);
	\begin{pgfonlayer}{bg}
		\draw[] (0,0) -- (0,7.5);
		\draw[] (3,0) -- (3,7.5);
		\draw[] (6,0) -- (6,7.5);
		\draw[] (9,0) -- (9,7.5);
		\draw[] (12,0) -- (12,7.5);
		\draw[] (15,0) -- (15,7.5);
		\draw[] (18,0) -- (18 ,7.5);
	\end{pgfonlayer}
	
\draw[thick] (2,6.5) -- node[above, commodityLabel]{$(\textcolor{rwthblue}{1}, \textcolor{rwthmagenta}{M})$} (4,6.5);
\draw[thick] (8,6.5) -- node[above, commodityLabel]{$(\textcolor{rwthblue}{1}, \textcolor{rwthmagenta}{M})$} (10,6.5);
\draw[thick] (14,6.5) -- node[above, commodityLabel]{$(\textcolor{rwthblue}{1}, \textcolor{rwthmagenta}{M})$} (16,6.5);
\draw[thick] (2, 0.5) -- node[pos =.15, above, commodityLabel]{$(\textcolor{rwthblue}{7}, \textcolor{rwthmagenta}{\tfrac{1}{7}})$} (10, 0.5);
\draw[thick] (2, 1.5) -- node[pos =.15, above, commodityLabel]{$(\textcolor{rwthblue}{6}, \textcolor{rwthmagenta}{\tfrac{1}{42}})$} (10, 1.5);
\draw[thick] (5, 2.5) -- node[pos =.9, above, commodityLabel]{$(\textcolor{rwthblue}{10}, \textcolor{rwthmagenta}{\tfrac{1}{10}})$} (16, 2.5);
\draw[thick] (5, 3.5) -- node[pos =.9, above, commodityLabel]{$(\textcolor{rwthblue}{9}, \textcolor{rwthmagenta}{\tfrac{1}{90}})$} (16, 3.5);
\draw[thick] (8, 4.5) -- node[pos =.2, above, commodityLabel]{$(\textcolor{rwthblue}{4}, \textcolor{rwthmagenta}{\tfrac{1}{4}})$} (13, 4.5);
\draw[thick] (8, 5.5) -- node[pos =.2, above, commodityLabel]{$(\textcolor{rwthblue}{3}, \textcolor{rwthmagenta}{\tfrac{1}{12}})$} (13, 5.5);
	
	\node[draw=none] at (1.5, -0.5) {$x_1$};
	\node[draw=none] at (4.5, -0.5) {$\bar{x_1}$};
	\node[draw=none] at (7.5, -0.5) {$x_2$};
	\node[draw=none] at (10.5, -0.5) {$\bar{x_2}$};
	\node[draw=none] at (13.5, -0.5) {$x_3$};
	\node[draw=none] at (16.5, -0.5) {$\bar{x_3}$};
\end{tikzpicture}
			}
			\caption{\ac{FZA}-instance constructed in the proof of Theorem~\ref{thm:NP-hardness_Paths} for $\varphi = (x_1 \vee \bar{x}_2) \wedge (\bar{x}_1 \vee \bar{x}_3) \wedge (x_2 \vee x_3)$. The parameters of each depicted commodity $i$ are given as $(\ub_{i}, \weight_i)$.}
			\label{fig:PathReduction-FullExample}
		\end{figure}
		
		\paragraph{Clause commodities.}
		In a final step, we add commodities for all the clauses $\clause_j = (l_j^{(1)}, l_j^{(2)})$. We assume, w.l.o.g., that
		the variable-assignment gadget corresponding to $l_j^{(1)}$ is to the left of the variable-assignment gadget corresponding $l_j^{(2)}$. 
		We now add commodities that yield a revenue of $1$ for each satisfied clause and a revenue of $0$ for each unsatisfied clause. Let $B$ be the number of blocks in the constructed graph between the two variable-assignment gadgets corresponding to $l_j^{(1)}$ and $l_j^{(2)}$. We add two commodities $\ell_j^{(1)}$ and $\ell_j^{(2)}$ for each clause, both spanning from part \blockC\ of the variable-assignment gadget corresponding to $l_j^{(1)}$ up to part \blockA\ of the variable-assignment gadget corresponding to $l_j^{(2)}$. We set $\ub_{\ell_j^{(1)}} = 3B+1$ and	$\weight_{\ell_j^{(1)}} = \tfrac{1}{3B+1}$, as well as $\ub_{\ell_j^{(2)}} = 3B$ and $\weight_{\ell_j^{(2)}} = \tfrac{1}{(3B)(3B+1)}$. We call these commodities \textit{clause commodities} of clause $j$. This construction is illustrated in Figure~\ref{fig:PathReduction-FullExample}.
		
		We now examine the revenue obtainable from these commodities.
		
		\begin{itemize}
			\item If neither of the two literals $l_j^{(1)}, l_j^{(2)}$ of $\clause_j$ is satisfied, the outermost edges of both commodities $\ell_j^{(1)}$ and $\ell_j^{(2)}$ are cut.
			As there are exactly three cuts in every variable-assignment gadget between the gadgets corresponding to $l_j^{(1)}$ and $l_j^{(2)}$, there are $3B+2$ cuts on $\pathGraph_{\ell_j^{(1)}}$ and $\pathGraph_{\ell_j^{(1)}}$, respectively, violating $\ub_{\ell_j^{(1)}}$ and $\ub_{\ell_j^{(2)}}$. As a result, both commodities drop out and yield a revenue of 0.
			\item If one of the two literals in $\clause_j$ is satisfied, exactly one of the outermost edges of $\pathGraph_{\ell_j^{(1)}}$ and $\pathGraph_{\ell_j^{(2)}}$ is cut. Thus, both paths contain $3B+1$ cuts, thereby respecting $\ub_{i_1}$ while violating $\ub_{i_2}$. The associated revenue is then $(3B+1) \weight_{\ell_j^{(1)}} = 1$.
			\item If both literals in $\clause_j$ are satisfied---analogous to the first two cases---$\pathGraph_{\ell_j^{(1)}}$ and $\pathGraph_{\ell_j^{(2)}}$  contain $3B$ cuts. Thus, we obtain revenue for both of them adding up to
			\begin{align*}
				&\hphantom{=}(3B)\bigg(\frac{1}{3B+1} + \frac{1}{(3B)(3B+1)}\bigg) = (3B)\frac{(3B+1)}{(3B)(3B+1)} = 1.
			\end{align*}
		\end{itemize}
		
		We conclude that for sufficiently large $M$, a variable assignment satisfying at least $y$ clauses of $\phi$ corresponds to a set of cuts in the constructed \ac{FZA} instance that yields at least $z = 29 M n + y$ revenue.
		
		Regarding the choice of $M$, we observe that despite every literal $l$ appearing in at most three clauses, the cuts in the variable-assignment gadget corresponding to $l$ can influence all commodities representing clauses. In particular, reducing the number of cuts in the gadget by one could increase the revenue from every pair of clause commodities by 1. Since there are no other commodities spanning the variable-assignment gadget of literal $l$ and its negation, choosing $M > m$ suffices. 
		This concludes the proof of Theorem~\ref{thm:NP-hardness_Paths}.
	\end{proof}
	
	
	\section{Parameterized Algorithms for \textsc{FZA} on Paths}
	\label{app:parameterized}
	In this section, we present parameterized algorithms for \ac{FZA} when restricted to paths. Recall the following definitions from \cref{sec:Introduction}.
	
	\begin{definition}[Parameters for \ac{FZA}]
		Let $\mathcal{I}$ be an instance of \ac{FZA}.
		\begin{enumerate}[a)]
			\item $\ub_{\max} \coloneqq \max_{i\in[k]} \ub_i$ is the maximum budget of any single commodity. 
			\item $|P|_{\max}  \coloneqq \max_{i\in[k]} |P_i|$ is the maximum path length of any single commodity.
			\item The \emph{congestion} of $\mathcal{I}$ is the maximum number of commodities passing through any single edge
			\[\operatorname{cong}(\mathcal{I}) \coloneqq \max_{e\in E}|\{i \in [k] \ | \ e\in P_i\}|.\]
		\end{enumerate}
	\end{definition}
	
	W.l.o.g., we may assume that $u_i \leq |P_i|$ for all commodities $i \in [k]$, and thus $u_{\max} \leq |P|_{\max}$. 
	Recall that a problem $\mathcal{P}$ with parameter $\ell$ is in $\XP$ if it can be solved in polynomial time for any fixed $\ell$, i.e., there is an algorithm for $\mathcal{P}$ with running time $\mathcal{O}(n^{f(\ell)})$ for some function $f$. $\mathcal{P}$ is in $\FPT$ with parameter $\ell$ if there is an algorithm for $\mathcal{P}$ with a running time of $\mathcal{O}(f(\ell)\cdot n^c)$ for some constant $c$. Clearly, any problem in $\FPT$ with parameter $\ell$ is also in $\XP$ for that parameter.
	
	In \cref{sec:umax_pmax}, we show that \ac{FZA} on paths is in $\XP$ when parameterized with $u_{\max}$. A very similar dynamic program shows that \ac{FZA} is in $\textsf{FPT}$ when parameterized with $|P|_{\max}$. 
	Subsequently, in \cref{sec:congestion}, we show that when parameterized with $\con$, \ac{FZA} on paths is in $\XP$, and when parameterized with $\max\{\con, u_{\max}\}$ it is in $\FPT$.
	Note that the construction we used to prove \cref{thm:APX-hardness_Stars} is constant in all three parameters introduced above. As a result, unless \textsf{P} $=$ \textsf{NP}, there is no hope to generalize the parameterized algorithms presented in this section to trees. 
	
	For all three dynamic programs described below, choose an arbitrary orientation of the path $T$ and denote its edges from left to right by $e^{(1)}, \ldots, e^{(n-1)}$. Further, we write $e' \prec e$ if $e'$ is to the left of $e$.
	
	\subsection{Dynamic Program for Parameters $u_{\max}$ and $|P|_{\max}$}
	\label{sec:umax_pmax}
	\paragraph{Parameter $u_{\max}$.} Assume that we are given an instance $\mathcal{I}$ of \ac{FZA} on a path $T$. For ease of notation we use $\ell = u_{\max}$. We define the subinstance of $\mathcal{I}$ up to edge $e^{(j)}$ as the subpath $P^{(j)}$ from $e^{(1)}$ to $e^{(j)}$ along with all commodities that share at least one edge with $P^{(j)}$. Further, we restrict the path of every such commodity $i$ to $P_i \cap P^{(j)}$. The remaining parameters $u_i$ and $w_i$ of $i$ remain unchanged.
	
	We compute a table consisting of entries $\DPTable(e^{(j)}, (c_1, \ldots, c_{\ell +1}))$. Each entry denotes the maximum revenue that can be obtained from a solution $F$ to the subinstance of $\mathcal{I}$ up to edge $e^{(j)}$ such that $c_1, \ldots, c_{\ell+1} \in P^{(j)}$ are the $\ell+1$ rightmost cuts in $F$ (with $c_1$ being the leftmost and $c_{\ell+1}$ being the rightmost of these cuts).
	We populate the table entries $\DPTable(e^{(j)}, \cdot)$ for $j = 1, \ldots, n-1$. For a cleaner initialization, we artificially append a path of length $\ell+1$ with edges $e^{(-\ell)}, \ldots, e^{(0)}$ to the left end of $T$. As there are no commodities whose path intersects either of these edges, we can assume that all these artificial edges are cut. Hence, we initialize
	\begin{alignat*}{3}
		&\DPTable(e^{(0)}, (e^{(-\ell)}, \ldots, e^{(0)})) &&= 0, \\
		&\DPTable(e^{(0)}, (c_1,\ldots, c_{\ell+1})) &&= -\infty, \qquad \text{otherwise.}
	\end{alignat*}
	Assume that all table entries for edge $e^{(j-1)}$ have already been computed for some $j \geq 1$. We now compute table entries $\DPTable(e^{(j)}, (c_1, \ldots , c_{\ell+1}))$.
	
	\emph{Case 1: $c_{\ell  +1} \neq e^{(j)}$.} Let $F$ be an optimal solution corresponding to $\DPTable(e^{(j-1)}, (c_1, \ldots , c_{\ell+1}))$. Since the solution corresponding to $\DPTable(e^{(j)}, (c_1, \ldots , c_{\ell+1}))$ cannot cut edge $e^{(j)}$, the cut set $F$ is still an optimal solution up to edge $e^{(j)}$ given that $(c_1,\ldots,c_{\ell+1})$ are its rightmost cuts. 
	Thus, to compute the new entry, it suffices to add the revenue obtained from all commodities whose path begins with edge $e^{(j)}$, that is
	\begin{equation}
		\label{eq:umax_recursive_no_cut}
		\DPTable(e^{(j)}, (c_1,\ldots, c_{\ell+1})) = \DPTable(e^{(j-1)}, (c_1,\ldots, c_{\ell+1})) + \sum_{\substack{i \in [k], \, e^{(j)} \text{ is the} \\ \text{leftmost edge of } P_i}} w_i \cdot f(0).
	\end{equation}
	
	\emph{Case 2: $c_{\ell  +1} = e^{(j)}$.} In this case, a corresponding solution must cut $e^{(j)}$. For all commodities whose path $P_i$ includes $e^{(j)}$, the position of all cuts left of $c_{1}$ is irrelevant, as either
	\begin{enumerate}[a)]
		\item the first cut on $P_i$ is $c_{1}$, or to the right of $c_{1}$, or
		\item all the cuts $c_1, \ldots, c_{\ell+1}$ are on $P_i$, in which case $i$ drops out since $u_i \leq \ell$.
	\end{enumerate}
	Thus, to obtain $\DPTable(e^{(j)}, (c_1, \ldots, c_{\ell}, e^{(j)}))$, we choose a maximum table entry among all $\DPTable(e^{(j-1)}, (c^*, c_1, \ldots, c_{\ell}))$ for $c^*$ left of $c_{1}$, and then factor in the change in revenue from cutting $e^{(j)}$. Thus, we obtain the following recursive formula
	\begin{align}
		\DPTable(e^{(j)}, (c_1,& \ldots, c_{\ell}, e^{(j)})) = \max_{c^* \prec c_{1}} \DPTable(e^{(j-1)}, (c^*, c_1,\ldots, c_{\ell})) \nonumber \\
		&+ \sum_{\substack{i \in [k], \, e^{(j)} \text{ is the} \\ \text{leftmost edge of } P_i}} w_i \cdot f(0) \label{eq:umax_recursive_cut}\\
		&+ \sum_{\substack{i \in [k], \, e^{(j)} \in P_i}} \rev{i}(e^{(j)} \ | \ \{c_1,\ldots,c_\ell\}). \nonumber
	\end{align}
	Here, $\rev{i}(e^{(j)} \ | \ \{c_1,\ldots,c_{\ell}\})$ denotes the marginal gain obtained by commodity $i$ from cutting $e^{(j)}$ given the set of previous cuts $\{c_1,\ldots, c_{\ell}\}$. Given the discussion above, this revenue is well-defined (that is, it does not depend on the position of cuts to the left of $c_{1}$). For $i \in [k]$ with $e^{(j)} \in P_i$, let $z_i = |P_i \cap \{c_1,\ldots,c_{\ell}\}|$. We compute
	\begin{equation}
		\label{eq:marginal_gain_param}
		\rev{i}(e^{(j)} \ | \ \{c_1,\ldots,c_{\ell}\}) = \begin{cases}
			w_i \cdot (f(z_i + 1) - f(z_i)), & z_i < u_i, \\
			-w_i \cdot f(z_i), & z_i = u_i, \\
			0, & z_i > u_i.
		\end{cases}
	\end{equation}
	Let $F$ be an arbitrary solution to the subinstance up to edge $e^{(j)}$ with $\{c_1,\ldots,c_{\ell}\} \subseteq F$. In the first case, $i$ gets served by $F$ and does not drop out upon adding $e^{(j)}$ to $F$. In the second case, $F$ serves $i$ but $F \cup \{e^{(j)}\}$ does not, and in the third case $i$ already dropped out under $F$ and still does so after adding $e^{(j)}$.
	
	Combining the two cases, we obtain the following result.
	\begin{theorem}
		\label{thm:param_umax}
		Let $\mathcal{I}$ be an instance of \ac{FZA} on a path $T$ and let $\ell \coloneqq u_{\max}$. The optimal revenue for $\mathcal{I}$ is given by
		\[\max_{(c_1,\ldots,c_{\ell+1}) \in E^{\ell+1}} \DPTable(e^{(n-1)}, (c_1,\ldots,c_{\ell+1})).\]
		This value can be calculated in time $n^{\mathcal{O}(\ell)}$. Thus, \ac{FZA} on paths is in $\XP$ when parameterized with $u_{\max}$.
	\end{theorem}
	\begin{proof}
		By the observations above, the maximum among all entries $\DPTable(e^{(n-1)}, \cdot)$ contains the optimal revenue. It remains to bound the running time. The table has $\mathcal{O}(n^{\ell+2})$ entries. To compute an entry falling into case 1 we need time $\mathcal{O}(k)$. In case 2, we first compute the maximum of $\mathcal{O}(n)$ table entries. Afterwards, we compute the revenue from all commodities featuring $e^{(j)}$ as their leftmost edge, which takes $\mathcal{O}(k)$ additional time. As $k \in \mathcal{O}(n^3)$, we get an overall running time of $n^{\mathcal{O}(\ell)}$.
	\end{proof}
	
	\paragraph{Parameter $|P|_{\max}$.} We now consider \ac{FZA} parameterized with $|P|_{\max}$. We highlight the similarities to the previous observations and use $\ell = |P|_{\max}$. We alter the definition of the table entries slightly. For edge $e^{(j)}$, we now compute $\DPTable(e^{(j)}, C)$ for each subset $C \subseteq \{e^{(j-\ell+1)}, \ldots , e^{(j-1)}, e^{(j)}\}$.
	We want $\DPTable(e^{(j)}, C)$ to contain the maximum revenue that can be obtained from the subinstance of $\mathcal{I}$ up to edge $e^{(j)}$ when cutting exactly the subset $C$ among all edges in $\{e^{(j-\ell+1)}, \ldots , e^{(j-1)}, e^{(j)}\}$. 
	We first alter $T$ by appending an artificial edge $e^{(0)}$ at its left endpoint, and initialize $\DPTable(e^{(0)}, \emptyset) = 0$.
	
	Let $i\in [k]$ be a commodity for which $e^{(j)} \in P_i$. As $|P_i| \leq \ell$, the path $P_i$ cannot contain $e^{(j-\ell)}$ or any edges left of it, and thus $\rev{i}(F)$ is uniquely determined by $C$. Hence, we obtain the following recursive formula for $\DPTable(e^{(j)}, C)$ in an analogous way to Equation~ \eqref{eq:umax_recursive_no_cut} and Equation~ \eqref{eq:umax_recursive_cut}. 
	For the case that $e^{(j)} \not\in C$ we obtain
	\[
	\DPTable(e^{(j)}, C) = \max \{\DPTable(e^{(j-1)}, C), \, \DPTable(e^{(j-1)}, C \cup \{e^{(j-\ell)}\})\} + \sum_{\substack{i \in [k], \, e^{(j)} \text{ is the} \\ \text{leftmost edge of } P_i}} w_i \cdot f(0),\]
	and for $e^{(j)} \in C$ we get
	\begin{align*}
		\DPTable(e^{(j)}, C) =& \max \Big\{\DPTable\big(e^{(j-1)}, C \setminus \{e^{(j)}\}\big), \DPTable\big(e^{(j-1)}, C \setminus \{e^{(j)}\} \cup \{e^{(j-\ell)}\}\big)\Big\}\\
		+& \sum_{\substack{i \in [k], \, e^{(j)} \text{ is the} \\ \text{leftmost edge of } P_i}} w_i \cdot f(0)\\
		+& \sum_{\substack{i \in [k], \, e^{(j)} \in P_i}} \rev{i}(e^{(j)} \ | \ (C \setminus \{e^{(j)}\})).
	\end{align*}
	Here, $\rev{i}(e^{(j)} \ | \ (C \setminus \{e^{(j)}\}))$ is defined analogously to $\rev{i}(e^{(j)} \ | \ \{c_1,\ldots,c_{\ell}\})$ in \cref{eq:marginal_gain_param}. This value is well-defined, as $|P_{\max}|$ is bounded by $\ell$. This yields the following result.
	\begin{corollary}
		\label{cor:param_pmax}
		\ac{FZA} on paths is in $\FPT$ when parameterized with $|P|_{\max}$.
	\end{corollary}
	\begin{proof}
		The optimal revenue can be computed by taking the maximum of all table entries $\DPTable(e^{(n-1)}, \cdot)$. This requires evaluating $\mathcal{O}(n\cdot2^{|P|_{\max}})$ table entries, each of which can be computed in time $\mathcal{O}(k)$. Since $k \in \mathcal{O}(n^3)$, the result follows.
	\end{proof}
	
	\subsection{Dynamic Programs Parameterized with Congestion}
	\label{sec:congestion}
	Again, let $\mathcal{I}$ be an instance of $\ac{FZA}$ on some underlying path $T$. We append an additional artificial edge $e^{(0)}$ to the left of $T$ and design a dynamic program which processes the edges of the path from left to right. Given some current edge $e^{(j)}$, we define the following three disjoint subsets of commodities for $j \in \{1,\ldots,n-1\}$
	\begin{align*}
		K_1(e^{(j)}) &\coloneqq \{i \in [k] \ | \ e^{(j-1)} \not\in P_i, e^{(j)} \in P_i\}, \\
		K_2(e^{(j)}) &\coloneqq \{i \in [k] \ | \ \{e^{(j-1)}, e^{(j)}\} \subseteq P_i\}, \\
		K_3(e^{(j)}) &\coloneqq \{i \in [k] \ | \ e^{(j-1)} \in P_i, e^{(j)} \not\in P_i\}.
	\end{align*}
	We call the commodities in  $K_1(e^{(j)})$ \emph{new} commodities, the commodities in $K_2(e^{(j)})$ \emph{active} commodities, and the commodities in $K_3(e^{(j)})$ \emph{finished} commodities w.r.t. $e^{(j)}$.
	
	We now want to compute table entries
	\[
	\DPTable(e^{(j)}, X),
	\quad x_i \in \{-\infty, -1, 0, 1, \ldots, u_i\} 
	\]
	Here, $X$ denotes a vector with entries $x_i$ for $i \in K_1(e^{(j)}) \cup K_2(e^{(j)})$.
	An entry $\DPTable(e^{(j)}, X)$ captures the maximum revenue that can be obtained in the subinstance up to edge $e^{(j)}$ among all solutions $F$ that satisfy $|F \cap P_i| = u_i - x_i$ for every $i \in K_1(e^{(j)}) \cup K_2(e^{(j)})$. (By slight abuse of notation, we say that this condition is satisfied for $x_i = -\infty$ if and only if $|F \cap P_i| > u_i +1$).
	We call $x_i$ the \emph{slack} of commodity $i$. To initialize the table, we set $\DPTable(e^{(0)}, ()) = 0$.
	
	We give recursive formulas for $\DPTable(e^{(j)}, X)$ assuming that all table entries $\DPTable(e^{(j-1)}, X)$ have been computed. We distinguish between the cases where there exist no new commodities (that is, $K_1(e^{(j)}) = \emptyset$) and where there exist new commodities.
	
	\emph{Case 1: $K_1(e^{(j)}) = \emptyset$.} In this case, we can either cut $e^{(j)}$ or not cut $e^{(j)}$. If we do not cut $e^{(j)}$, we look up the best table entry $\DPTable(e^{(j-1)}, Y)$ where $Y$ is a vector defined on $K_2(e^{(j)}) \cup K_3(e^{(j)})$, and $Y|_{K_2(e^{(j)})} = X|_{K_2(e^{(j)})}$. Here, $Y|_{K_2(e^{(j)})}$ and $X|_{K_2(e^{(j)})}$ denote the vectors $X$ and $Y$ restricted to the indices in $K_2(e^{(j)})$, respectively. The remaining entries $Y|_{K_3(e^{(j)})}$ can be chosen arbitrarily with $y_i \in \{-\infty, -1, 0, 1, \ldots, u_i\}$. If these conditions are satisfied, we call \emph{$Y$ feasible for $X$ without cut}.
	
	If we do cut $e^{(j)}$, we only need to consider table entries $\DPTable(e^{(j-1)}, Y)$ for which the slack of all commodities in $K_2(e^{(j)})$ is one less in $X$ than in $Y$. That is, for all $i \in K_2(e^{(j)})$ we require
	\begin{align*}
		y_i = x_i + 1, &\quad x_i \in \{-1,0,1,\ldots,u_i\}, \\
		y_i \in\{-\infty, -1\}, &\quad x_i = -\infty.
	\end{align*}
	Again, we can choose all values $y_i$ for $i \in K_3(e^{(j)})$ arbitrarily. If these conditions are met, we call \emph{$Y$ feasible for $X$ with cut}. Overall, we obtain the following recursive formula
	\begin{align}
		\DPTable(e^{(j)}, X) = & \max\Big\{ \max_{\substack{Y \text{ feasible for }X \\ \text{without cut}}}\DPTable(e^{(j-1)}, Y),\nonumber\\
		& \max_{\substack{Y \text{ feasible for }X \\ \text{with cut}}}\DPTable(e^{(j-1)}, Y) \label{eq:param_cong_recursive}\\
		& + \sum_{i \in K_2(e^{(j)})} \rev{i} (e^{(j)} \ | \ x_i)\Big\}\nonumber
	\end{align}
	Here, $\rev{i}(e^{(j)} \ | \ x_i)$ denotes the marginal gain from commodity $i$ obtained by cutting $e^{(j)}$ after which $i$ has a slack of $x_i$. More concretely,
	\begin{equation}
		\rev{i}(e^{(j)} \ | \ x_i) = \begin{cases}
			w_i \cdot (f(u_i - x_i) - f(u_i - x_i - 1)), &x_i \geq 0, \\
			-w_i \cdot f(u_i), & x_i = -1, \\
			0, & x_i = -\infty.
		\end{cases}
		\label{eq:param_cong_marginal_gain}
	\end{equation}
	In the last case, the marginal gain is $0$ as commodity $i$ already dropped out prior to cutting $e^{(j)}$.
	
	\emph{Case 2: $K_1(e^{(j)}) \neq \emptyset$.} This case works similarly to case 1, with the difference that $X|_{K_1(e^{(j)})}$ already determines whether edge $e^{(j)}$ is cut. Clearly if edge $e^{(j)}$ is not cut, we must have $x_i = u_i$ for all $i \in K_1(e^{(j)})$, and if $e^{(j)}$ is cut we must have $y_i = u_i - 1$ for all $i \in K_1(e^{(j)})$. All other values for  $X|_{K_1(e^{(j)})}$ are clearly infeasible, and thus we set the corresponding table entries to $-\infty$. We consider the two feasible sub-cases individually.
	
	\emph{Case 2.1: $x_i = u_i$ for all $i \in K_1(e^{(j)})$.}
	In this case, we obtain exactly the same formula as for the first option in \cref{eq:param_cong_recursive}, but we additionally need to account for the base revenue obtained from the new commodities. Thus, we obtain
	\[\DPTable(e^{(j)}, X) = \max_{\substack{Y \text{ feasible for }X \\ \text{without cut}}}\DPTable(e^{(j-1)}, Y) + \sum_{i \in K_1(e^{(j)})} w_i \cdot f(0).\]
	
	\indent\emph{Case 2.2: $x_i = u_i-1$ for all $i \in K_1(e^{(j)})$.}
	In this case, we get a similar formula to the second option in \cref{eq:param_cong_recursive}, however, we need to add the base revenue obtained from all new commodities with $u_i \geq 1$. Hence,
	\begin{align*}
		\DPTable(e^{(j)}, X) = \max_{\substack{Y \text{ feasible for }X \\ \text{with cut}}}\DPTable(e^{(j-1)}, Y) & + \sum_{i \in K_2(e^{(j)})} \rev{i} (e^{(j)} \ | \ x_i) \\
		&+ \sum_{\substack{i \in K_1(e^{(j)}), \\ u_i \geq 1}} w_i \cdot f(1)
	\end{align*}
	where the marginal gain $\rev{i} (e^{(j)} \ | \ x_i)$ is defined exactly the same as in \cref{eq:param_cong_marginal_gain}.
	
	Combining the two cases, we conclude that \ac{FZA} on paths is in $\XP$ when parameterized with $\con$.
	\begin{theorem}
		\label{thm:param_cong}
		Let $\mathcal{I}$ be an instance of \ac{FZA} on a path. The optimal revenue for $\mathcal{I}$ is given by
		\[\max_{X} \ \DPTable(e^{(n-1)}, X)\]
		where the maximum is taken over all vectors $X$ with $x_i \in \{-\infty, -1, 0, 1, \ldots, u_i\}$ for all $i \in K_1(e^{(n-1)}) \ \cup \ K_2(e^{(n-1)})$. This value can be calculated in time $n^{\mathcal{O}(\con)}$. Thus, \ac{FZA} on paths is in $\XP$ when parameterized with $\con$.
	\end{theorem}
	\begin{proof}
		For some fixed $e^{(j)}$, we need to compute at most $\prod_{i \in K_1(e^{(j)}) \cup K_2(e^{(j)})} (u_i+3) = n^{\mathcal{O}(\con)}$ table entries, as clearly we have $u_i + 3 \in \mathcal{O}(n)$ for all $i$. To compute a single table entry, we need to choose the best table entry for all possible values of $Y|_{K_3(e^{(j)})}$, and compute the marginal gain made by the new edge $e^{(j)}$. For the first part, we need to consider $n^{\mathcal{O}(\con)}$ many entries, as $|K_3(e^{(j)})| \leq \con$. Computing the marginal gain made by the new cut is possible in time $\operatorname{poly}(n)$. Thus, we obtain an overall running time of $n^{\mathcal{O}(\con)}$.
	\end{proof}
	
	\paragraph{Parameter $\max\{\con, u_{\max}\}$.} Assume that $\max\{\con, u_{\max}\} = \ell$. We use exactly the same dynamic program as for parameter $\con$ and analyze its runtime for the new parameter $\ell$. Note that for every fixed $j$, we only need to compute $\mathcal{O}\left((\ell+3)^{\ell}\right)$ table entries (as there are at most $\ell+3$ possible values that each $y_i$ can take). Similarly, when computing some fixed table entry $\DPTable(e^{(j)}, X)$ there are only $\mathcal{O}\left((\ell + 3)^{\ell}\right)$ possible values for $Y|_{K_3(e^{(j)})}$ that we need to consider. Thus, computing any single table entry now only takes time $\mathcal{O}\left((\ell+3)^{\ell} + \operatorname{poly}(n)\right)$. Overall, we obtain the following, stronger result for this parameter.
	
	\begin{corollary}
		\label{cor:param_cong_umax}
		\ac{FZA} on paths is in $\FPT$ when parameterized with $\ell \coloneqq \max\{\con, u_{\max}\}$.
	\end{corollary}
	
	
	\section{A PTAS for Paths}
	\label{app:ptas}
	In this section, we briefly sketch how the PTAS due to Grandoni and Rothvoss \cite{GrandoniRothvossPTAS} for the \textsc{Highway} problem (i.e., \textsc{Tollbooth} on a path) can be transformed to a PTAS for \ac{FZA} on a path with minor modifications. We neither provide a full description nor analysis of the algorithm, but instead focus on the single step where the algorithm for \ac{FZA} differs from the one for \textsc{Highway}. Note that in \cref{thm:APX-hardness_Stars} we show \textsf{APX}-hardness of \ac{FZA} when restricted to stars, and thus there is no hope to extend this PTAS beyond paths. We first give a general overview of the algorithm and provide further details afterwards. Recall that in a solution to \textsc{Tollbooth}, we determine a price $p_e$ for each edge $e$. Given some subpath $T' = (V', E')$ and a (partial) solution $p \colon E' \to \mathbb{R}_{+}$, we call $\sum_{e \in E'} p_e$ the \emph{price} of path $T'$.
	
	\paragraph{Overview of the PTAS by Grandoni and Rothvoss \cite{GrandoniRothvossPTAS}.}
	First, the instance of \textsc{Tollbooth} is reduced to a so-called \emph{well-rounded} instance that admits an optimal solution with edge prices in $\{0,1\}$. Let $\varepsilon > 0$. Given a well-rounded instance, the algorithm is then split up into a \emph{bounding phase}, a \emph{dynamic programming phase}, and a \emph{scaling phase.}
	\begin{itemize}
		\item In the bounding phase, first the total price of an optimal solution $m^*$ is guessed (which is possible since the instance is well-rounded). Then, the original path $T$ gets enclosed into a bounding path of length $(1/\varepsilon)^y$ for some randomly chosen $y$. We denote the new path by $T_0$.
		\item In the dynamic programming phase, for every subpath $T' \subseteq T_0$, we compute the optimal revenue $R(T', m)$ that can be obtained by assigning price $1$ to exactly $m$ edges in $T'$. We first do this by brute force for $m = (1/\varepsilon)^y$, and then use dynamic programming for other values of $m$. The optimal solution value is then given by $R(T_0, m^*)$.
		\item In the scaling phase, the computed solution is scaled down by a factor of $\delta/(\delta+2)$ for $\delta = 1/(2\varepsilon)$.
	\end{itemize}
	
	\paragraph{Adapting the PTAS for \ac{FZA}.} Until the final step, every edge gets a price of either $0$ or $1$. This can be interpreted as not cutting or cutting the edge in the context of \ac{FZA}. Hence, it suffices to revise the final scaling phase where prices get multiplied with the non-integer factor of $\delta/(\delta+2)$ to obtain the PTAS for \ac{FZA}.
	
	In \textsc{Highway}, the scaling step is needed to deal with ``long'' commodities. When computing $R(T_0,m^*)$, let $(T_1, \ldots, T_\gamma)$ be the best partition of $T_0$ in $\gamma$ subpaths (for some $\gamma$ depending on $\varepsilon$) such that each subpath has a price of $m^*/\gamma$. (The algorithm simply evaluates all possible partitions and chooses the best one). Consider some commodity $i$ served by $\OPT$ for which $P_i$ crosses $\Omega(1/\varepsilon)$ many of the subpaths $T_j$. On each of the inner segments, the solution on the inner segments of $T'$ assigns price $1$ to as many edges (that is, makes as many cuts) in $P_i$ as the optimal solution. However, on the two outer subpaths $T_1$ and $T_{\gamma}$, it might happen that the algorithm assigns a higher price to $T_1 \cap P_i$ and $T_{\gamma} \cap P_i$ than the optimal solution. This might lead to commodity $i$ dropping out. However, by slightly scaling down the solution, we can prevent $i$ from dropping out without losing a lot of revenue on shorter commodities.
	
	\paragraph{The thinning phase.} To adapt the PTAS to \ac{FZA}, we substitute the scaling phase by a \emph{thinning phase} as follows. Consider each $T_j$ in the best partition of $T_0$ into $\gamma$ many subpaths each with a price of $m^*/\gamma$, and let 
	\[\varphi \coloneqq \Big\lceil \frac{2}{\delta+2} \cdot \frac{m^*}{\gamma}\Big\rceil.\]
	Let $F_j = \{e_1,\ldots,e_{|F_j|}\}$ be the cut set on path $P_j$ before the thinning phase. We construct $|F_j| = m^*/\gamma$ candidate ``thinned-out'' solutions
	\[F_j^{(\ell)} \coloneqq \{e_{\ell}, e_{\ell+1}, \ldots, e_{\ell+|F_j|-\varphi}\}, \quad \ell\in\{1,\ldots,|F_j|\}\]
	(using addition modulo $|F_j|$ for the indices). Each of these candidate solutions consists of $|F_j| - \varphi$ consecutive cuts, and every edge is part of exactly $|F_j| - \varphi$ many thinned-out sets $F_j^{(\ell)}$. From these solutions, let $F_j^*$ be the one that obtains the maximum revenue from all commodities fully contained in $T_j$. Note that for $F_j^*$ the following two properties hold.
	\begin{itemize}
		\item $|F_j^*| \leq \frac{\delta}{\delta+2}|F_j|$.
		\item Let $M_j$ be the set of commodities fully contained in $T_j$. Then,
		\[\rev{M_j}(F_j^*) \geq (1-\mathcal{O}(\varepsilon))\rev{M_j}(F_j).\]
	\end{itemize}
	
	The first property ensures that there is no ``long'' commodity that gets served by $F^*$ but not by $F_{\ALG} = \bigcup_{j=1}^{\gamma} F_j^*$. The second property ensures that we do not lose too much revenue generated by \emph{short} commodities that are contained within a single subpath $T_j$. To show the second property we use the probabilistic method. For the sake of the analysis, suppose that we choose $\ell$ uniformly at random (instead of enumerating all candidate solutions and choosing the best). Then, every cut in $F_j$ is contained in $F_j^{(\ell)}$ with the same probability of $1-\mathcal{O}(\varepsilon)$. Thus, by the probabilistic method, there must exist some good value of $\ell$ for which $F_j^{(\ell)}$ loses only a factor of $(1-\mathcal{O}(\varepsilon))$ on all commodities fully contained in $F_j$ when compared to $\OPT$. The remaining proof of the approximation guarantee follows closely along the lines of the original analysis by Grandoni and Rothvoss \cite{GrandoniRothvossPTAS}.
	
	
	\section{Single Density Proofs}
	\label{app:SingleDensityProofs}
	
	\subsection{Proof of Theorem~\ref{thm:SingleDensity}}
	Recall the following definitions from Section~\ref{sec:SingleDensity}:
	We define the density $d_i$ of commodity $i\in[k]$ as $d_i \coloneqq \frac{u_i}{|\pathGraph_i|} \in \{0\} \cup [\tfrac{1}{n-1},1]$ and, based on that, partition the set of commodities into density classes 
	\begin{equation}
		\label{app-eq:DensityClassDefinition}
		M_j \coloneqq \bigl\{i \in [k] \,\big\vert\, \ub_i \geq 1 \text{ and } d_i \in (2^{-j}, 2^{1-j}]\bigr\}
	\end{equation}
	for each $j \in \{1, \dots, \lceil\log_2(n)\rceil\}$ and $M_0 \coloneqq \{i \in [k] \,\big\vert\, \ub_i =0\}.$ We fix some root $r \in V$ and define the candidate solutions
	\[\cutset_{j, \eStart} = \{e\in E \,\vert\, \operatorname{dist}(r,e) \equiv \eStart \mod 2^{j+1}\}.\]
	To obtain the sets $\cutset_{j, \eStart}^{\text{(rand)}}$, we then iterate through all edges in $\cutset_{j, \eStart}$ and delete each edge independently at random with probability $1/2$.
	
	\SingleDensity*
	\begin{proof}
		For this entire proof, let $\cutset^*$ be an optimal solution to the instance in question. Recall that for every $j \in \{1, \dots, \lceil\log_2(n)\rceil\}$ we need to show that the Single Density algorithm tests a set $\cutset_j$ satisfying
		\begin{align*}
			\mathbb{E}[\rev{M_j}(\cutset_j)] \geq \alpha \cdot \rev{M_j}(\cutset^*)
		\end{align*}
		for some constant $\alpha$ independent of $j$. First note that for class $M_0$ the set $\cutset_0 = \emptyset$ satisfies $\rev{M_0}(\cutset_0) \geq \rev{M_0}(\cutset^*)$. Thus, from now on, we restrict our analysis to the sets $\cutset_{j,\eStart}$ for $j>0$ evaluated by the algorithm.
		
		\paragraph{Commodities with $\ub_i \geq 2$.}
		Fix an arbitrary $j \in \{1,\ldots,\lceil\log_2(n)\rceil\}$ and consider only the commodities in $M_j$ with $\ub_i \geq 2$. 
		Let $i$ be one such commodity with path $\pathGraph_i\subseteq E$. Let $v = \argmin_{v \in \pathGraph_i} \text{dist}(r,v)$, i.e., the vertex in $P_i$ closest to $r$. If $v$ is an endpoint of $\pathGraph_i$, we leave $\pathGraph_i$ unchanged. Otherwise, we split $\pathGraph_i$ into two subpaths $\pathGraph_i^{(1)}$ and $\pathGraph_i^{(2)}$ going from either endpoint of $\pathGraph_i$ to $v$. For all values $\eStart \in \{0, \ldots, 2^{j+1}-1\}$ it holds that
		\begin{align}
			|\cutset_{j,\eStart} \cap \pathGraph_i| &= |\cutset_{j,\eStart} \cap \pathGraph_i^{(1)}| + |\cutset_{j,\eStart} \cap \pathGraph_i^{(2)}| \nonumber \\
			&\leq \bigg\lceil \frac{|\pathGraph_i^{(1)}|}{2^{j+1}} \bigg\rceil + \bigg\lceil \frac{|\pathGraph_i^{(2)}|}{2^{j+1}} \bigg\rceil \leq \bigg\lceil \frac{|\pathGraph_i|}{2^{j+1}} \bigg\rceil + 1 \leq \frac{\ub_i}{2} + 1 \leq \ub_i,
			\label{eq:singledensity}
		\end{align}
		where the last inequality holds restricted to commodities with $\ub_i \geq 2$. Hence, no matter the choice of $\eStart \in \{0, \dots, 2^{j+1}-1\}$, the solution $\cutset_{j, \eStart}$ serves commodity $i \in M_j$.
		We employ the probabilistic method to show the existence of some ``good'' $F_{j,\eStart}^{(\text{rand})}$, i.e., some $F_{j,\eStart}^{(\text{rand})}$ for which
		\[\mathbb{E}[\rev{M_j}(F_{j,\eStart}^{(\text{rand})})] \geq \alpha \cdot \rev{M_j}(F^*).\]
		We choose $\eStart \in \{0, \ldots, 2^{j+1}-1\}$ uniformly at random. For every edge $e$, the probability that $e$ is contained in $F_{j,\eStart}$ is exactly $2^{-(j+1)}$. Thus, summing over all edges in $P_i$ we obtain
		\begin{equation}
			\label{eq:single_density_ui}
			\mathbb{E}[|F_{j,\eStart} \cap P_i|] = \frac{|P_i|}{2^{j+1}} \geq \frac{u_i}{4},
		\end{equation}
		where the last inequality holds as $i \in M_j$ and thus, $d_i \leq 2^{1-j}$ which implies $|\pathGraph_i| > 2^{j-1}$.
		
		However, lower bounding $\mathbb{E}[|F_{j,\eStart} \cap P_i|]$ does not directly yield a lower bound on $\mathbb{E}[\rev{i}(\cutset_{j, \eStart})]$ as we only assume that the pricing function $f$ is subadditive but not necessarily linear. To obtain a lower bound for $\mathbb{E}[\rev{i}(\cutset_{j, \eStart})]$ observe that, for any pair $\eStart_1, \eStart_2 \in \{0,\ldots,2^{j+1}-1\}$
		\begin{equation}
			\label{eq:single_density_expected_cuts}
			\Big| \big|F_{j,\eStart_1} \cap P_i| - |F_{j,\eStart_2} \cap P_i\big|\Big| \leq 2.
		\end{equation}
		This is because $F_{j,\eStart_1}$ and $F_{j,\eStart_2}$ place cuts in an alternating fashion on $P_i^{(1)}$ and $P_i^{(2)}$, respectively. We introduce a random variable $X_i \coloneqq |F_{j,\eStart} \cap P_i|$ (when randomizing over $\eStart$) and distinguish two cases.
		
		\emph{Case 1: $X_i \geq 1$.} By \cref{eq:single_density_expected_cuts}, the minimum and maximum value that $X_i$ can take under different realizations of $\eStart$ differs by a factor of at most 3. Here, the worst case arises when $X_i$ can take values $1$, $2$ or $3$. 
		By \cref{eq:single_density_ui}, there must be some $\eStart$ such that $|F_{j,\eStart} \cap P_i| \geq \frac{u_i}{4}$. Together with the above consideration we get 
		$|F_{j,\eStart} \cap P_i| \geq \frac{1}{3} \cdot \frac{u_i}{4}$ for all $\eStart \in \{0,\ldots,2^{j+1}-1\}.$ As commodity $i$ always gets served, and $f$ is subadditive, this implies
		\begin{equation}
			\label{eq:single_density_case1}
			\mathbb{E}[\rev{i}(F_{j, \eStart})] \geq \frac{1}{12}\rev{i}(F^*).
		\end{equation}
		
		\emph{Case 2: There exists some $\eStart$ for which $F_{j,\eStart} \cap P_i = \emptyset$.} In this case,\cref{eq:single_density_expected_cuts} yields that $X_i$ can only take on values in $\{0,1,2\}$. Thus, it holds that $\mathbb{E}[X_i] < 2$. 
		By \cref{eq:single_density_ui} and subadditivity
		\[
		\rev{i}(F^*) \leq w_i\cdot f(u_i) \leq w_i \cdot f(4\mathbb{E}[X_i]) \leq w_i \cdot 4\mathbb{E}[X_i]\cdot f(1).
		\]
		We consider $\mathbb{E}[\rev{i}(F_{j,\eStart} \cap P_i)]$. Clearly, with probability at least $\mathbb{E}[X_i]/2$, the algorithm places at least one cut on $P_i$, and in this case the revenue obtained from $i$ is at least $w_i \cdot f(1)$. Thus, we conclude
		\[
		\mathbb{E}[\rev{i}(F_{j,\eStart})] \geq \frac{\mathbb{E}[X_i]}{2}\cdot w_i \cdot f(1) \geq \frac{1}{8}\rev{i}(F^*).
		\]
		
		Taking the worse guarantee from the two subcases, we obtain $\mathbb{E}[\rev{i}(F_{j, \eStart})] \geq \tfrac{1}{12}\rev{i}(F^*)$. Thus, after deleting every edge with probability $1/2$,
		\[\mathbb{E}[\rev{i}(F_{j, \eStart}^{(\text{rand})})] \geq \frac{1}{24}\rev{i}(F^*).\]
		
		\paragraph{Commodities with $u_i = 1$.}
		Now, we turn to the analysis of commodities with $\ub_i = 1$. Again, fix some $j \in \bigl\{1,\dots, \lceil\log_2(n)\rceil\bigr\}$ and consider a single commodity $i \in M_j$ (with $u_i =1$).
		
		For every $e' \in \pathGraph_i$, the probability $\text{Pr}[e' \in \cutset_{j,\eStart}]$ that $\cutset_{j,\eStart}$ contains $e'$ is $2^{-(j+1)}$. Summing over all edges of $\pathGraph_i$ we get
		\[\mathbb{E}[|\pathGraph_i \cap \cutset_{j,\eStart}|] \geq 2^{-(j+1)} \cdot |\pathGraph_i|.\]
		As the random variable $|\pathGraph_i \cap \cutset_{j,\eStart}|$ can only take values in $\{0,1,2\}$,
		\[\text{Pr}[|\pathGraph_i \cap \cutset_{j,\eStart}| \geq 1] \geq 2^{-(j+2)} \cdot |\pathGraph_i|\]
		Observing that $|\pathGraph_i| > 2^{j-1}$ as $u_i = 1$ and $d_i \in (2^{-j}, 2^{1-j}]$, we get
		\[\text{Pr}[|\pathGraph_i \cap \cutset_{j,\eStart}| \geq 1] \geq \frac{1}{8}.\]
		Every edge in $\cutset_{j,\eStart}$ survives the deletion step with probability $\frac{1}{2}$, hence
		\[\text{Pr}[|\pathGraph_i \cap \cutset_{j,\eStart}^{\text{(rand)}}| = 1] \geq \frac{1}{16}.\]
		Because $u_i = 1$, this bound directly translates to the revenue that $F_{j, \eStart}^{(\text{rand})}$ obtains from $i$, that is $\mathbb{E}[\rev{i}(F_{j, \eStart}^{(\text{rand})})] \geq \tfrac{1}{16}\rev{i}(F^*)$.
		
		\paragraph{Combining the two cases.} Considering all commodities $i\in M_j$ again, by linearity of expectation,
		\[\mathbb{E}[\rev{M_j}(F_{j, \eStart}^{(\text{rand})})] \geq \frac{1}{24}\rev{M_j}(F^*).\]
		Let $F_j$ be the best solution from among the candidate solutions $F_{j,\eStart}^{(\text{rand})}$. By the probabilistic method, $\cutset_j$ satisfies
		\[
		\mathbb{E}\bigl[\rev{M_j}(\cutset_j)\bigr] \geq \frac{1}{24} \rev{M_j}(\cutset^*).
		\]
		Since the algorithm returns the best solution $\cutset_{\ALG}$ among all $\lceil \log_2(n)\rceil + 1$ classes, we get
		\[
		\mathbb{E}\bigl[\rev{}(\cutset_{\ALG})\bigr] \geq \frac{1}{24} \cdot \frac{1}{\lceil \log(n) \rceil + 1} \cdot  \rev{}(\cutset^*).
		\]
		Regarding the runtime, the Single Density algorithm constructs $\mathcal{O}(n \cdot \log n)$ solution sets. For each of these it evaluates the revenue which is clearly possible in time $\mathcal{O}(k \cdot n)$, when using adjacency lists.
	\end{proof}
	
	\subsection{Deterministic Variants of the Single Density Algorithm for Special Cases}
	\label{sec:det_variants_single_density}
	In this subsection, we consider two special cases of \ac{FZA} and sketch how the Single Density Algorithm can be altered to (a) lose a smaller constant factor in the approximation guarantee, and (b) avoid using randomization. We begin with the special case of $T$ being a path and then consider the special case where the pricing function $f$ satisfies $f(0) > 0$.
	
	\paragraph{\ac{FZA} on paths.} Let $T$ be a path. We denote the edges of $T$ by $e^{(1)},\ldots,e^{(n-1)}$. We choose the partition of the commodities exactly in the same way as in the previous section, that is, we define
	\begin{align*}
		M_j &= \{i \in [k] \ | \ u_i \geq 1 \text{ and }d_i \in (2^{-j}, 2^{1-j}]\}, \quad j \geq 1,\\
		M_0 &= \{i \in [k] \ | \ u_i = 0\}.
	\end{align*}
	Again, we let $\cutset_0 = \emptyset$. For $j \in [\lceil \log_2 n\rceil]$ and $\theta \in \{1,\ldots,2^{j}\}$, let
	\[\cutset_{j, \eStart} \coloneqq \{e^{(\eStart')} \ | \ \eStart' \equiv \eStart \mod 2^j\},\]
	that is, $\cutset_{j, \eStart}$ includes every $2^j$-th edge of $T$ starting with $e^{(\theta)}$. We choose $\cutset_j$ to be the best of the candidate solutions $\cutset_{j,\theta}$.
	We show the following lemma, which directly implies a deterministic $4\cdot (\lceil \log_2(n) \rceil + 1)$-approximation on \ac{FZA} on paths.
	\begin{corollary}
		\label{cor:single_density_paths}
		Assume that we are given an instance of \ac{FZA} on a path. For the cut set $F_j$ constructed as outlined above, it holds that
		\[\rev{M_j}(F_j) \geq \frac{1}{4}\cdot \rev{M_j}(F^*).\]
	\end{corollary}
	\begin{proof}
		The proof goes along the same lines as the proof of \cref{thm:SingleDensity}. For $j=0$ the statement is trivial. Let $j \in [\lceil \log_2 n \rceil] $ and let $i \in M_j$ be arbitrary. We choose $\eStart \in \{1,\ldots, 2^j\}$ uniformly at random. Note, that no matter the choice of $\eStart$, the solution $\cutset_{j, \eStart}$ always serves all commodities in $M_j$. We obtain
		\[\mathbb{E}[|F_{j,\eStart} \cap P_i|] \geq \frac{u_i}{2}.\]
		We instead consider $\mathbb{E}[\rev{i}(F_{j,\eStart})]$ and argue in a very similar way to the proof of \cref{thm:SingleDensity} by considering the random variable $X_i \coloneqq |F_{j, \theta} \cap P_i|$. However, on paths the values that $X_i$ can assume for different realizations of $\theta$ can differ by at most one. If $X_i \geq 1$, the worst case arises if the two possible values are $1$ or $2$, and thus we only lose a factor of 2 (rather than 3 as in \cref{eq:single_density_case1}) in the approximation ratio. Since this case was the bottleneck of our previous analysis, we obtain
		\[\mathbb{E}[\rev{i}(F_{j,\eStart})] \geq \frac{1}{4}\rev{i}(F^*).\]
		Summing over all commodities $i \in M_j$ using the probabilistic method shows the existence of some $\eStart$ such that for $F_j = F_{j, \eStart}$ we have
		\[
		\rev{M_j}(F_j) \geq \frac{1}{4}\rev{M_j}(F^*).
		\]
	\end{proof}
	
	\paragraph{\ac{FZA} with $f(0) > 0.$} 
	We now consider \ac{FZA} on trees with $f(0) > 0$. Recall, that the reason for which we resorted to randomization in the proof of \cref{thm:SingleDensity} was to prevent commodities with $u_i = 1$ from dropping out. 
	However, for $f(0) > 0$, the candidate solution $F_0 = \emptyset$ already suffices to obtain a constant fraction of the revenue that an optimal solution $F^*$ can generate from these commodities. We define
	\begin{align*}
		M_j &= \{i \in [k] \ | \ u_i \geq 2 \text{ and }d_i \in (2^{-j}, 2^{1-j}]\},  \quad j \geq 1,\\
		M_0 &= \{i \in [k] \ | \ u_i \in \{0,1\}\}. 
	\end{align*}
	To construct the candidate solutions $F_j$, we let $F_0 = \emptyset$ and for $j>0$ define the sets $F_{j,\eStart}$ as in the proof of \cref{thm:SingleDensity}
	\[
	\cutset_{j, \eStart} = \{e\in E \,\vert\, \operatorname{dist}(r,e) \equiv \theta \mod 2^{j+1}\}.
	\]
	We then choose $F_j$ as the best solution among all $F_{j,\eStart}$ (omitting the randomized deletion of edges) and obtain the following guarantee.
	\begin{corollary}
		Assume that we are given an instance of \ac{FZA} where the pricing function fulfills $f(0) > 0$. For the solution $F_j$ described above it holds that
		\[\rev{M_j}(F_j) \geq \min\bigg\{\frac{f(0)}{f(1)}, \frac{1}{12}\bigg\}\cdot \rev{M_j}(F^*).\]
	\end{corollary}
	\begin{proof}
		Consider $i\in M_0$. For $u_i = 0$, clearly we have $\rev{i}(F_0) \geq \rev{i}(F^*)$. For $u_i = 1$ it is possible that $F^*$ makes one cut on $\pathGraph_i$ whereas $F_0$ does not. In this case, we have
		\[\rev{i}(F_0) = \frac{f(0)}{f(1)}\rev{i}(F^*).\]
		By summing over all commodities in $M_0$, we obtain
		\[\rev{M_0}(F_0) \geq \frac{f(0)}{f(1)}\rev{M_0}(F^*).\]
		
		For all $i \in M_j$ for some $j > 0$ we proceed analogously to the proof of \cref{thm:SingleDensity}. However, since $u_i \geq 2$, every candidate solution $F_{j, \eStart}$ serves every $i \in M_j$. Thus, we can omit the random deletion of edges.
	\end{proof}
	
	\subsection{A Simplified Algorithm}
	In this section we briefly sketch a simpler version of the Single Density Algorithm that still yields an approximation ratio $\mathcal{O}(\log(n))$. \cref{alg:single_density_simplified} describes the algorithm in pseudocode. While \cref{alg:single_density_simplified} is simpler to describe and even a coarse analysis suffices to show the same (asymptotic) approximation guarantee it heavily relies on randomization. As a consequence, the techniques used to de-randomize the algorithm when $G$ is a path or for $f(0) > 1$ (cf. Section~\ref{sec:det_variants_single_density}) do not apply to Algorithm~\ref{alg:single_density_simplified}. Also, this algorithm loses (with our very coarse analysis) much larger constants in the approximation guarantee.
	
	\cref{alg:single_density_simplified} still partitions commodities based on their density. However, instead of iterating over offsets $\eStart$ and choosing the best solution among all sets $F_{j, \eStart}^{\text{(rand)}}$, \cref{alg:single_density_simplified} constructs the candidate solution $F_j$ for every $j \in \{1, \dots, \lceil\log_2(n)\rceil\}$ by including each edge with probability $2^{-j-1}$ independently at random.
	
	\begin{algorithm}[htbp]
		\algorithmInit
		\Indm
		\KwIn{A tree $T=(V,E)$ and commodities $(\pathGraph_i, \weight_i, \ub_i)_{i \in [k]}$}
		\KwOut{A set $F_{\ALG} \subseteq E$ of edges to cut.}
		\Indp
		Compute densities $(d_i)_{i\in[k]}$ and sets of commodities $M_j$\;
		Let $F_0 = \emptyset$\;
		\For{$j = 1, \ldots, \lceil \log_2(n)\rceil$}
		{
			Include each edge in $F_j$ with probability $2^{-j-1}$ independently at random
		}
		\KwRet{the best solution among all $F_j$ for $j \in \{0, \ldots, \lceil \log_2(n) \rceil\}$}
		\caption{Simplified Single Density approximation for \ac{FZA}}
		\label{alg:single_density_simplified}
	\end{algorithm}
	
	\begin{restatable}{theorem}{SimplifiedSingleDensity}
		\label{thm:SimplifiedSingleDensity}
		\cref{alg:single_density_simplified} computes an $\mathcal{O}(\log n)$-approximation for \ac{FZA} in expectation.  
	\end{restatable}
	
	\begin{proof}
		Let $i \in M_j$ be arbitrary and define $X$ as the random variable $|P_i \cap F_j|$. We distinguish two cases and show that in either case the candidate solution $F_j$ (in expectation) generates a constant fraction of the revenue generated by an optimal solution $F^*$. 
		
		\bigskip
		
		\noindent \emph{Case 1: $\mathbb{E}[X] \geq 2$.} In this case, we use Markov- and Chernoff-bounds to show that with constant probability, (i) commodity $i$ gets served, and (ii) we have $X \geq c \cdot u_i$ for some suitable constant $c$. Note, that $\mathbb{E}[X] \leq \tfrac{u_i}{2}$. Thus, we can apply a Markov bound to see that
		\[\text{Pr}(i \text{ drops out}) \leq \text{Pr}(X \geq 2\mathbb{E}[X]) \leq \frac{1}{2}.\]
		To show that $X$ is sufficiently large with constant probability, first note that $X$ follows a binomial distribution, in particular, $X \sim \text{Bin}(|P_i|, 2^{-j-1})$. 
		We use the Chernoff bound 
		\[
		\text{Pr}\bigl(X \leq (1-\delta)\mathbb{E}[X]\bigr) \leq e^{-\frac{\delta^2\mathbb{E}[X]}{2}}.
		\]
		Using $\mathbb{E}[X] \geq 2$ and choosing $\delta = 0.9$ this implies
		\[
		\text{Pr}\left(X \leq \frac{\mathbb{E}[X]}{10}\right) \leq 0.45.
		\]
		By union bound, we can see that with probability at least $1/20$, both (i) and (ii) hold. If this is the case, we have $\mathbb{E}[\rev{i}(F_j)] \geq \frac{1}{4\cdot 10\cdot 20}\rev{i}(F^*) = \frac{1}{800}\rev{i}(F^*)$. 
		
		\bigskip
		
		\noindent \emph{Case 2: $\mathbb{E}[X] < 2$.} Note that $d_i \leq 2^{1-j}$ for all $i \in M_j$. Thus, we obtain $|P_i| \geq 2^{j-1}\cdot u_i \geq 2^{j-1}$. On the other hand, $d_i > 2^{-j}$ which implies that $|P_i| < 2^j\cdot u_i \leq 2^{j+3}$ where the last inequality uses the fact that $\mathbb{E}[X] \geq \tfrac{u_i}{4}$ and thus $u_i \leq 8$. We calculate (using $X \sim \text{Bin}(|P_i|, 2^{-j-1})$)
		\[\text{Pr}(X = 1) = \frac{|P_i|}{2^{j+1}} \left(1-\frac{1}{2^{j+1}}\right)^{|P_i|-1} \hspace*{-1.5mm} \geq \frac{1}{4}\left(1-\frac{1}{2^{j+1}}\right)^{|P_i|-1} \hspace*{-1.5mm} \geq  \frac{1}{4}\left(1-\frac{1}{2^{j+1}}\right)^{2^{j+3}}.\]
		Here, the sequence $\left(1-\frac{1}{2^{j+1}}\right)^{2^{j+3}}$ is increasing and converges to $e^{-4}$. Thus, the worst case occurs for $j=1$ in which case above calculation yields $\text{Pr}(X = 1) \geq \tfrac{1}{400}$. Combining this with $u_i \leq 8$ we conclude that $\mathbb{E}[\rev{i}(F_j)] \geq \tfrac{1}{3200}\rev{i}(F^*)$.
	\end{proof}
	
	
	\section{Missing Proofs in \cref{sec:SublogApprox}}
	\label{app:Sublog_AppendixProofs}
	
	\subsection{Non-skeleton algorithm---Proof of \cref{l.non-skeleton}}
	\label{app:non-skel}
	Recall the algorithm for computing a cut set on the non-skeleton edges. We contract the entire skeleton $\mathcal{S}$ into a single vertex $v_{\mathcal{S}}$. The contracted tree decomposes into subtrees $T_1, \ldots, T_{d'}$ each rooted at $v_{\mathcal{S}}$. Let $T_j, j \in [d']$, be one of these subtrees. With probability $1/2$ we let $F_j = \emptyset$ and say that $T_j$ is \emph{inactive}. Otherwise, we use \cref{thm:SingleSource} to solve a rooted instance $\mathcal{I}(T_j)$ rooted at $v_{\mathcal{S}}$. This instance includes those commodities for which one endpoint is an inner vertex of $T_j$ and the other endpoint is either in $\mathcal{S}$ or in some inactive subtree. We denote the solution computed this way by $F_j$, and let $F^{E\setminus\mathcal{S}} \coloneqq \cup_{j \in [d']} F_j$.
	\nonskelApproximation*
	\begin{proof}
		Let $i \in [k]$ be some commodity that has exactly one of its endpoints in subtree $T_j$. We call $i$ \emph{nice} for $T_j$ if $T_j$ is active and either of the following two conditions holds:
		\begin{itemize}
			\item The other endpoint belongs to $\mathcal{S}$, or
			\item the other endpoint belongs to an inactive subtree, and at least half of the edges in $(F^*\setminus{\mathcal{S}})\cap P_i$ lie in $T_j$.
		\end{itemize}
		We denote the set of nice commodities for $T_j$ by $M_{\text{nice}, j}$. Note that every $i \in M_{\text{nice}, j}$ is contained in the set of commodities of the auxiliary rooted instance $\mathcal{I}(T_j)$, which we denote by $M_{\mathcal{I}(T_j)}$. Let $i \in [k]$ be some commodity for which at least one of its endpoints is outside of $\mathcal{S}$. With probability $1/2$, exactly one of the (at most two) trees $T_j$ containing the endpoints of $P_i$ is active, and again with probability $1/2$ the active subtree is the one for which $F^*$ makes more cuts on $T_j \cap P_i$. Thus, every commodity is nice (for some tree $T_j$) with probability at least $1/4$.
		
		As the DP from \cref{thm:SingleSource} solves rooted instances to optimality and we have $M_{\text{nice},j} \subseteq M_{\mathcal{I}(T_j)}$. Hence, the revenue obtained from $\cutset_j$ satisfies $\rev{M_{\mathcal{I}(T_j)}}(\cutset_j) \geq \rev{M_{\text{nice}, j}}(F^* \cap  T_j)$. As for all $i \in M_{\text{nice}, j}$ at least half of the edges in $(F^*\setminus{\mathcal{S}})\cap P_i$ belong to $T_j$, and the pricing function $f$ is subadditive, at least half of the revenue obtained from commodities in $M_{\text{nice}, j}$ on non-skeleton edges is obtained on $T_j$. This leads to the following inequality.
		\begin{equation}
			\label{eq:Non-SkeletonCaseAuxRevenue}
			\rev{M_{\mathcal{I}(T_j)}}(\cutset_j) \geq \rev{M_{\text{nice}, j}}(F^* \cap T_j) \geq \frac{1}{2}\rev{M_{\text{nice},j}}(F^*\setminus{\mathcal{S}}). \hfill
		\end{equation}
		
		Let $F^{E\setminus \mathcal{S}} = \bigcup_{j \in [d']} F_j$ be the union of all computed partial solutions $F_j$ over all active subtrees $T_j$. If $i \in M_{\mathcal{I}(T_j)}$, the set $F^{E\setminus \mathcal{S}}$ contains no edge in $P_i \setminus T_j$, and thus $\rev{i}(F^{E\setminus \mathcal{S}}) = \rev{i}(F_j)$. Note that the sets $M_{\mathcal{I}(T_j)}$ for different subtrees $j$ are pairwise disjoint. Using this and Equation~\eqref{eq:Non-SkeletonCaseAuxRevenue}, we conclude
		\[
		\rev{}(F^{E\setminus\mathcal{S}}) 
		\geq \rev{\cup_{j \in [d']} M_{\mathcal{I}(T_j)}}(F^{E\setminus\mathcal{S}}) 
		= \sum_{j \in [d']} \rev{M_{\mathcal{I}(T_j)}}(F_j) 
		\overset{(\ref{eq:Non-SkeletonCaseAuxRevenue})}{\geq} \frac{1}{2}\sum_{j\in[d']} \rev{M_{\text{nice}, j}}(F^* \setminus \mathcal{S}).
		\]
		As every commodity is nice (for some subtree $T_j$) with probability at least $1/4$, and every nice commodity intersects with exactly one active subtree $T_j$, we obtain
		\[\mathbb{E}\Big[\frac{1}{2}\sum_{j\in[d']} \rev{M_{\text{nice}, j}}(F^* \setminus \mathcal{S})\Big] \geq \frac{1}{8}\rev{}(F^* \setminus \mathcal{S}).\]
		This concludes the proof.
	\end{proof}
	
	
	\subsection{Generalized Rooted FZA on a Path}
	\label{app:GeneralizedRooted}
	
	In \cref{subsec:skeleton} we introduced the problem variant \emph{generalized rooted \ac{FZA} on a path} (GR-FZA). Solving instances of this problem is a fundamental component of the sublogarithmic approximation algorithm. Recall that compared to a regular rooted instance of \ac{FZA} an instance of GR-FZA features commodity-specific subadditive revenue functions $f_i\colon \mathbb{N} \to \mathbb{R}_{\geq 0}$, $i\in [k]$ instead of one global revenue function $f$. 
	The following lemma shows that we can compute a revenue-maximizing solution $\cutset$ to a generalized rooted instance on a path subject to the condition that $|\cutset| = y$ for a fixed solution size $y\in \{1,\ldots,|E|\}$.
	
	\begin{proposition}
		\label{prop:SingleSourceForSublog}
		Let $\mathcal{I}$ be an instance of \textsc{GR-FZA} with commodities $i\in [k]$, each paired with a commodity-specific non-decreasing, subadditive pricing  function $f_i\colon \mathbb{Z}_{\geq 0} \to \mathbb{R}_{\geq 0}$. Further, let $y \in \{0, \dots, |E|\}$,
		An optimal solution $\cutset^*$ to $\mathcal{I}$ among all solutions $\cutset$ with $|\cutset| = y$ can be computed in time $\mathcal{O}(n^2k)$.
	\end{proposition}
	\begin{proof}
		Let $T=(v_1, \ldots, v_t)$  denote the underlying path of $\mathcal{I}$ with root $r = v_1$.
		We denote the endpoint different from $r$ of every commodity $i$ as $t_i$. For each $x\in \{0,1, \ldots, y\}$, and every vertex $v_j$ let $R_{v_j}(x,y)$ denote the maximum revenue which can be obtained solely from commodities $i\in [k]$ satisfying $v_j\in \pathGraph_i$, and under the two restrictions that (i) in total exactly $y$ edges are cut, and (ii) there are exactly $x$ cuts on the subpath between $r$ and $v_j$.
		We initialize $R_{v_j}(x,y)=-\infty$ whenever $x > j-1$ (as then it is impossible to cut $x$ edges on the subpath from $r$ to $v_j$), and whenever $x>y$.
		Further, we set
		\[
		R_{v_t}(y,y) = \sum_{i\in [k]}\bigl\{\weight_i \cdot f(y) \,\big\vert\, \sink_i = v_t,\ y\le \ub_i \bigr\},
		\]
		and $R_{v_t}(x,y)=-\infty$ for all $x<y.$
		For $j=t$ down to $1$, and all  $x\le y$ with $x\le j-1$, we can compute $R_{v_j}(x,y)$ by the following recursive formula
		\[
		R_{v_j}(x,y) = \sum_{i\in [k]}\bigl\{\weight_i \cdot f(x) \,\big\vert\, \sink_i = v_j,\ x\le \ub_i \bigr\}
		+ \max\bigl\{R_{v_{j+1}}(x,y),\ R_{v_{j+1}}(x+1,y) \bigr\}.
		\]
		Additionally, we maintain a subset $\cutset_{v_j,x,y}$ of edges of the subpath between $v_j$ and $v_t$, populated along the procedure above. 
		That is, we initialize $\cutset_{v_t,x,y}=\emptyset$ for $x\le y$, and for any $j<t$, we set
		$\cutset_{v_j,x,y} = \cutset_{v_{j+1},x,y}$ if the maximum in the formula above is attained in the first term, and
		$\cutset_{v_j,x,y} = \cutset_{v_{j+1},x+1,y} \cup \{\{v_j,v_{j+1}\}\}$ otherwise.
		It follows from construction that $\cutset^*(y)\coloneqq \cutset_{v_1} (0,y)$ corresponds to a solution of size exactly $y$ with a revenue of
		$R_{v_1}(0,y)$ such that $\cutset^*(y)$ yields maximum revenue among all solutions of size $y$.
	\end{proof}
	
	\subsection{Algorithm for Skeleton Edges---Proof of Lemma~\ref{lem:sublog_submodular_2}}
	\label{app:skeleton}
	
	\sublogSubadditivity*
	\begin{proof}
		Consider the subset $M \subseteq [k]$ of commodities served by $F$, and let $F'$ be any subset of $F$ with $|F'| = \lceil \tfrac{|F|}{2} \rceil$. Clearly, all commodities in $M$ get served by both $F'$ and $F \setminus F'$, and by the subadditivity of $f$ for all $i \in M$ it holds that
		\[
		\rev{i}(F') + \rev{i}(F \setminus F') \geq \rev{i}(F).
		\]
		We sum over all commodities in $M$ and obtain
		\[\rev{M}(F') + \rev{M}(F \setminus F') \geq \rev{M}(F) = \rev{}(F).\]
		Thus, the statement follows when choosing $\bar{F}$ as the better solution among $F'$ and $F \setminus F'$. (If $|\bar{F}| < 2^{\lfloor \log_2 m\rfloor}$ we add arbitrary edges from $F$ until $\bar{F}$ has the desired size. This can only increase $\rev{M}(\bar{F})$).
	\end{proof}
\end{document}